\documentclass[a4paper,USenglish,cleveref,autoref,thm-restate]{lipics-v2021}
\usepackage[dvipsnames]{xcolor}
\usepackage[most]{tcolorbox}
\usepackage{lipsum}  
\usepackage[utf8]{inputenc}
\usepackage{tcolorbox}
\usepackage{amssymb,amsmath,amsthm}
\usepackage{xspace}
\usepackage{todonotes}
\usepackage{colortbl}

\usepackage{tcolorbox}
\tcbuselibrary{skins} 
\usepackage{xcolor}

\nolinenumbers

\usepackage{mathtools}
\DeclarePairedDelimiter\ceil{\lceil}{\rceil}

\usepackage[T1]{fontenc}
\definecolor{anti-flashwhite}{rgb}{0.95, 0.95, 0.96}

\newcommand{\sfast}{{\sc{Subset-FAST}}\xspace}
\newcommand{\sfastr}{{\sc{Subset-FAST-Rev}}\xspace}
\newcommand{\csfast}{{\sc{Colored Subset-FAST}}\xspace}

\newcommand{\OO}{{\mathcal{O}}\xspace}
\newcommand{\FF}{{\mathcal{F}}\xspace}

\newcommand{\INN}{{\sf IN}\xspace}
\newcommand{\OUTT}{{\sf OUT}\xspace}
\newcommand{\relevant}{{\texttt {Relevant}}\xspace}
\newcommand{\irrelevant}{{\texttt {Irrelevant}}\xspace}

\DeclareMathOperator{\sfas}{Sfas}

\newcommand{\CC}{{\mathcal{C}}}

\newcommand{\type}{{\mathtt{Type}}\xspace}

\newcommand{\rev}{{\mathsf{rev}}}
\newcommand{\arc}{{\mathsf{arc}}}
\newcommand{\scc}{{\textsf{SCC}}\xspace}
\newcommand{\indeg}{{\textsf{in-deg}}\xspace}
\newcommand{\spclorder}{{\mathsf{S} \text{-}\textsf{topological ordering}}\xspace}
\newcommand{\spclorderi}{{\mathsf{S_i} \text{-}\textsf{topological ordering}}\xspace}

\newcommand{\no}{\textsf{No}\xspace} 
\newcommand{\yes}{\textsf{Yes}\xspace}
\newcommand{\eqv}{\texttt{EqvCls}\xspace}
\newcommand{\surearc}{\texttt{sure arc}\xspace}
\newcommand{\surearcs}{\texttt{sure arcs}\xspace}

\newtheorem{reduction rule}{Reduction Rule}
\newtheorem*{reduction rule*}{Reduction Rule}

\definecolor{tyred}{rgb}{0.8, 0.0, 0.0}
\colorlet{mix}{red!50!black}

\providecommand{\keywords}[1]
{
	\textbf{\textit{Keywords---}} #1
}
\usepackage{hyperref}
\hypersetup{
	pdfnewwindow=true,      
	colorlinks=true,       
	linkcolor=mix,          
	citecolor=blue,        
	urlcolor=blue          
}

\newcommand{\defparprob}[4]{
	\vspace{1mm}
	\noindent\fbox{
		\begin{minipage}{0.96\textwidth}
			\begin{tabular*}{\textwidth}{@{\extracolsep{\fill}}lr} #1  & {\bf{Parameter:}} #3 \\ \end{tabular*}
			{\bf{Input:}} #2  \\
			{\bf{Task:}} #4
		\end{minipage}
	}
	\vspace{1mm}
}

\usepackage{algorithm}
\usepackage{algorithmicx}
\usepackage{algpseudocode}
\usepackage{mdframed}

\algtext*{EndWhile}
\algtext*{EndIf}
\algtext*{EndFor}

\title{A Quadratic Vertex Kernel and a Subexponential Algorithm for {\sc Subset-FAST}} 

\titlerunning{A Quadratic Vertex Kernel and a Subexponential Algorithm for Subset  FAST} 

\author{Satyabrata Jana}{University of Warwick, Coventry, UK}{satyamtma@gmail.com}{https://orcid.org/0000-0002-7046-0091}{Supported by the Engineering and Physical Sciences Research Council (EPSRC)
	via the project MULTIPROCESS (grant no. EP/V044621/1)}

\author{Lawqueen Kanesh}
{Indian Institute of Technology, Jodhpur, India}{lawqueen@iitj.ac.in}{
	https://orcid.org/0000-0001-9274-4119}{}

\author{Madhumita Kundu}{University of Bergen, Norway}{kundumadhumita.134@gmail.com}{https://orcid.org/0000-0002-8562-946X}{}

\author{Daniel Lokshtanov}{University of California, Santa Barbara, CA, USA}{daniello@ucsb.edu}{https://orcid.org/0000-0002-3166-9212}{}

\author{Saket Saurabh}{The Institute of Mathematical Sciences, HBNI, Chennai, India  \and University of Bergen, Norway }{saket@imsc.res.in}{https://orcid.org/0000-0001-7847-6402}{Supported by the European Research Council (ERC) under the European Union's Horizon 2020 research and innovation programme (grant agreement No. 819416); and he also acknowledges the support of Swarnajayanti Fellowship grant DST/SJF/MSA-01/2017-18.}

\authorrunning{S.~Jana, L.~Kanesh, M.~Kundu, D.~Lokshtanov, and S.~Saurabh} 

\Copyright{Jane Open Access and Joan R. Public} 



\category{} 

\relatedversion{} 

\acknowledgements{I want to thank \dots}

\nolinenumbers 

\EventEditors{John Q. Open and Joan R. Access}
\EventNoEds{2}
\EventLongTitle{42nd Conference on Very Important Topics (CVIT 2016)}
\EventShortTitle{CVIT 2016}
\EventAcronym{CVIT}
\EventYear{2016}
\EventDate{December 24--27, 2016}
\EventLocation{Little Whinging, United Kingdom}
\EventLogo{}
\SeriesVolume{42}
\ArticleNo{23}

\begin{document}
	
	\maketitle

	\begin{abstract}
		In the {\sc Subset Feedback Arc Set in Tournaments} (\sfast) problem we are given as input a tournament $T$ with a vertex set $V(T)$ and an arc set $A(T)$, along with a terminal set $S \subseteq V(T)$, and  an integer $ k$. The objective is to determine whether there exists a set $ F \subseteq A(T) $ with $|F| \leq k$ such that the resulting graph $T-F $ contains no cycle that includes any vertex of $S$. When $S=V(T)$ this is the classic {\sc Feedback Arc Set in Tournaments (\sc FAST)} problem.  
		We obtain the first polynomial kernel for this problem parameterized by the solution size. More precisely, we obtain an algorithm that, given an input instance $(T, S, k)$, produces an equivalent instance $(T',S',k')$ with $k'\leq k$ and $V(T')=\OO(k^2)$.
		
		It was known that {\sc FAST} admits a simple quadratic vertex kernel and a non-trivial linear vertex kernel. However, no such kernel was previously known for \sfast. 
		Our kernel employs variants of the most well-known reduction rules for {\sc FAST} and introduces two new reduction rules to identify irrelevant vertices. 
		As a result of our kernelization, we also obtain the first sub-exponential time {\sf FPT} algorithm for \sfast.

	\end{abstract}

	\section{Introduction}\label{sec:intro}

	In this paper we focus on a subset variant of the well-studied {\sc Directed Feedback Arc Set ({\sc DFAS})} problem in tournaments. In {\sc DFAS}, the input consists of a directed graph \(D\) and a positive integer \(k\). The goal is to determine whether there exists a set \(F \subseteq A(D)\) of size at most \(k\) such that \(D - F\) forms a directed acyclic graph (DAG). In the subset variant of this problem, referred to as {\sc Subset-DFAS}, we are additionally given a subset \(S \subseteq V(D)\) (known as a {\em terminal set}) and aim to determine whether there exists a set \(F \subseteq A(D)\) of size at most \(k\) such that \(D - F\) contains no cycles that include any vertex from \(S\). Notably, when \(S = V(D)\), {\sc Subset-DFAS} coincides with {\sc DFAS}.
	
	If we substitute arcs with vertices in the definitions above, we obtain the problems of {\sc Directed Feedback Vertex Set ({\sc DFVS})} and {\sc Subset-DFVS}. All four problems mentioned are known to be NP-complete~\cite{DBLP:books/fm/GareyJ79} and admit a polynomial-time approximation algorithm with a factor of \(\OO(\log n \log \log n)\) \cite{DBLP:conf/ipco/EvenNSS95}. Additionally, they have fixed-parameter tractable (FPT) algorithms with running times of \(2^{\OO(k^3)} n^{\OO(1)}\) \cite{DBLP:journals/talg/ChitnisCHM15,DBLP:journals/jacm/ChenLLOR08}. However, it remains unknown whether any of {\sc DFAS}, {\sc Subset-DFAS}, {\sc DFVS}, or {\sc Subset-DFVS} admit a constant factor approximation, a \(c^k n^{\OO(1)}\) time FPT algorithm, or a polynomial kernel. 
	The lack of progress on these problems has motivated research into their special instances.

	The algorithmic complexity of \textsc{DFAS} and \textsc{DFVS} changes significantly when the input is limited to tournaments—directed graphs where there is precisely one directed arc between each pair of vertices, or equivalently, orientations of complete undirected graphs. In fact, the problems \textsc{Feedback Arc Set in Tournament ({\sc FAST})} and \textsc{Feedback Vertex Set in Tournament ({\sc FVST})} are well-established problems in their own right. Both problems are known to be NP-complete \cite{DBLP:journals/siamdm/Alon06,DBLP:journals/cpc/CharbitTY07,DBLP:conf/wg/Stamm90}. While \textsc{FAST} has a polynomial-time approximation scheme \cite{DBLP:conf/stoc/Kenyon-MathieuS07}, \textsc{FVST} has a polynomial-time 2-approximation algorithm, which is optimal \cite{DBLP:journals/talg/LokshtanovMMPPS21}. In the context of parameterized complexity, \textsc{FAST} offers a subexponential time algorithm with a runtime of \(2^{\mathcal{O}(\sqrt{k})} n^{\mathcal{O}(1)}\) \cite{DBLP:conf/icalp/AlonLS09,DBLP:journals/corr/abs-0911-5094,DBLP:conf/esa/FominP13}, whereas \textsc{FVST} has an algorithm that runs in \(1.618^k n^{\mathcal{O}(1)}\) time \cite{DBLP:conf/stacs/KumarL16}. Concerning kernelization, \textsc{FAST} admits a linear vertex kernel \cite{DBLP:journals/jcss/BessyFGPPST11}, whereas \textsc{FVST}   possess  an almost linear  vertex kernel \cite{DBLP:conf/soda/BessyBTW23,DBLP:journals/jcss/PaulPT16}. Furthermore, it's crucial to note that solving \textsc{FAST} has significant implications across various applications, including ranking systems, sports competitions, voting systems, and preference aggregation \cite{DBLP:journals/4or/CharonH07,DBLP:journals/jair/CohenSS99,DBLP:conf/www/DworkKNS01,dwork2001rank}. In these scenarios, addressing cyclic inconsistencies is vital for generating a meaningful linear ordering.
	
	A natural question that has been explored recently is what happens when we consider the generalization of \textsc{FAST} and \textsc{FVST} to \textsc{Subset-FAST} and \textsc{Subset-FVST}, respectively. Can we extend the results of the non-subset variants to their subset variants? It is not difficult to show that both of these problems admit algorithms with a running time of \(3^k n^{\mathcal{O}(1)}\) by leveraging their connection to the \textsc{3-Hitting Set} problem. Bai and Xiao \cite{DBLP:journals/tcs/BaiX23} made significant improvements by designing a \(2^{k + o(k)}  n^{\mathcal{O}(1)}\) time algorithm for \textsc{Subset-FVST}. More recently, Jana et al. \cite{DBLP:conf/iwpec/JanaKK024} developed an algorithm for \textsc{Subset-FVST} with the same running time as that of \textsc{FVST}. Specifically, the algorithm runs in \(1.618^k n^{\mathcal{O}(1)}\) time. Furthermore, it is possible to construct a kernel with \(\mathcal{O}(k^2)\) vertices for \textsc{Subset-FVST} by utilizing its relationship with the \textsc{3-Hitting Set} problem.
	
	
	In this paper, we  consider \sfast and derive results that go beyond what can be attained through its connection to the {\sc $3$-Hitting Set} problem. Formally, we will consider the following problem.

		


	
	\defparprob{{\sc Subset Feedback Arc Set in Tournaments} (\sfast)}{A tournament $T=(V, A)$, a terminal set $S \subseteq V(T)$, and  an    integer $ k$.}{$k$}{Find  $ F \subseteq A $ with $|F| \leq k$ such that   $ T-F $ has no cycle containing a vertex of  $S$. }
	
	\subsection{Our Results and Methods}
	In this paper we obtain the first polynomial kernel for \sfast  parameterized by the solution size. More precisely, we obtain a polynomial time algorithm that, given an input instance $(T, S, k)$, produces an equivalent instance $(T',S',k')$ with $k'\leq k$ and $V(T')=\OO(k^2)$. In particular, we obtain the following result.

	\begin{restatable}{theorem}{faskernel} \label{theo:fastkernel}
		\sfast admits  a kernel with  $ \OO(k^2) $ vertices.     
	\end{restatable}
	
	\begin{figure}[t]
		\begin{center}
			\includegraphics[scale=0.65]{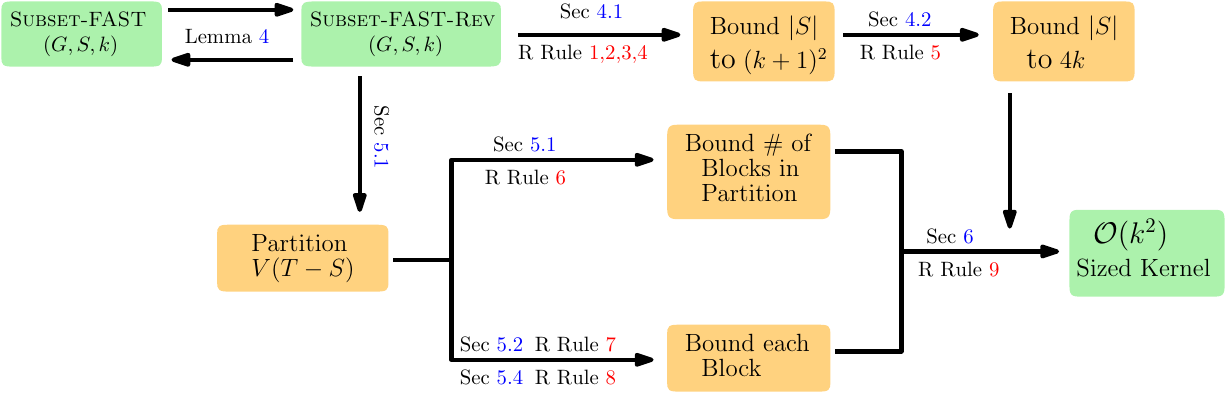}
		\end{center}
		\caption{A summary of the steps of our kernelization algorithm.} 
		\label{fig:flowchart}
	\end{figure}
	
	
	A schematic diagram showing the main steps and reduction rules of our algorithm is shown in Figure~\ref{fig:flowchart}.  We first show that the problem of deleting $k$ arcs to eliminate all directed $S$-cycles (cycle containing some vertices from $S$) is equivalent to the problem of reversing $k$ arcs to achieve the same outcome (\Cref{lem:minimal_equivalence}). In particular our problem is equivalent to the following reversal problem. Given a directed graph $G$ and a subset $F \subseteq A(G)$ of arcs, we denote $G \circledast F $ to be the directed graph obtained from $G$ by reversing all the arcs of $F$. 
	

	\defparprob{{\sc Subset Feedback Arc Set Reversal }(\sfastr)}{A tournament $T=(V(T),A(T))$, a vertex set $S \subseteq V(T)$,  and  an    integer $ k$.}{$k$}{Find  $ F \subseteq A(T)$ with $|F| \leq k$ such that   $ T\circledast F $ has no $S$-cycle.}

	\Cref{lem:minimal_equivalence} allows us to focus on the problem \sfastr. Next, we establish a bound on the number of terminal vertices \(S\) (\Cref{subsec:bound_terminals}). A straightforward quadratic bound can be derived by applying reduction rules similar to those used in obtaining a quadratic kernel for \textsc{FAST} \cite{DBLP:books/sp/CyganFKLMPPS15}. Specifically, we implement two simple reduction rules: (i) remove a vertex \(v\) from \(S\) that does not participate in any triangles, and (ii) if an arc appears in more than \(k\) distinct \(S\)-triangles, we reverse the arc and decrease the budget \(k\) by \(1\) (classic {\em sunflower reduction rule}). To achieve a linear bound on the size of \(S\), we closely follow the kernelization algorithm for \textsc{FAST} as outlined by Paul et al. \cite{DBLP:journals/jcss/PaulPT16}, and also by Bessy et al. \cite{DBLP:journals/jcss/BessyFGPPST11}. In particular, we utilize the \emph{conflict packing} technique and adapt the notion of a \emph{safe partition} to secure a linear bound on the terminals.
	
	After bounding the vertices in \(S\) by \(\OO(k)\), the algorithm shifts its focus to bounding the non-\(S\) vertices in \cref{subsec:partition_nont_eqvcls}. To accomplish this, we first partition the non-\(S\) vertices into equivalence classes based on their neighborhoods in \(S\). Specifically, all non-\(S\) vertices sharing the same set of in-neighbors and out-neighbors in \(S\) belong to the same equivalence class. A loose upper bound on the number of such equivalence classes is \(2^{|S|}\). However, if the input instance is a \yes instance, the number of equivalence classes is bounded by \(\OO(k)\).
	
	To support this observation, we introduce the concept of a \(\spclorder\) (\cref{def:ordered_part}). Intuitively, a \(\spclorder\) \(\Pi\) with respect to \(S\) consists of an ordered partition \(V_1 \uplus V_2 \uplus \ldots \uplus V_{\ell}\) of the vertices of \(T\), where parts containing \(S\)-vertices are singletons, and the parts that do not contain \(S\)-vertices form a strongly connected component in \(T\). Furthermore, for every pair \(i < j \in [\ell]\), all vertices of \(V_i\) precede all vertices of \(V_j\) in \(\Pi\), and every arc \((u,v)\) exists in \(T\), where \(u \in V_i\) and \(v \in V_j\). 
	
	We demonstrate that a tournament \(T\) is \(S\)-acyclic if and only if it admits such a \(\spclorder\). It follows that if \(T\) is an \(S\)-acyclic tournament, then each part \(V_i\) in the \(\spclorder\) that contains only non-\(S\) vertices represents a neighborhood equivalence class defined based on neighbors in \(S\). This further implies that the number of equivalence classes in an \(S\)-acyclic tournament is bounded by \(|S| + 1\). Consequently, in the case of a \yes instance, we can have at most \(k\) additional equivalence classes in \(T\). Therefore, we obtain a bound of \(|S| + k + 1 = \OO(k)\) on the number of equivalence classes in any tournament.
	
	Finally, we concentrate on each equivalence class, denoted by \(Z\). In Sections \ref{subsec:quadratic_bound} and \ref{subsec:linear_bound_eqclasses}, we demonstrate that if the size of the equivalence class is large, we can identify a significant subset of vertices, referred to as \(Z_{\irrelevant}\), such that ``most'' of the arcs incident to it can never be part of the solution (we call such arcs \surearcs). Furthermore, the only arcs that could potentially belong to \(F\) are those adjacent to at most \(\OO(k^2)\) vertices, which we designate as \(Z_{\relevant}\).
	In a sense, the vertices in \(Z_{\relevant}\) form a rigid skeleton (as discussed in \cref{subsec:quadratic_bound}). We can establish that if we refrain from removing the vertices in \(Z_{\relevant}\), then a vertex in \(Z_{\irrelevant}\) can be safely removed.
	
	
	One reason why the size of \( Z_{\relevant} = \mathcal{O}(k^2) \) is that it contains in-neighbors and out-neighbors of all vertices of $Z$ that have either an in-degree at most $k$ or an out-degree at most $k$ in $T[Z]$. It is known that the number of such vertices are upper bounded by $\OO(k)$ and hence the size of \( Z_{\relevant} = \mathcal{O}(k^2) \). We keep all these in-neighbors and out-neighbors in \( Z_{\relevant} \) because they validate the redundancy of the vertices in \( Z_{\irrelevant} \). However, the redundancy of a vertex in \( Z_{\irrelevant} \) depends on the \emph{number of in-neighbors} and \emph{number of out-neighbors} in \( Z_{\relevant} \), not on the actual neighbors themselves. We leverage this observation and apply an ``arc swapping argument'' (see \Cref{fig:swap} for an illustration) to obtain an equivalent instance in which, for any vertex in \( Z_{\irrelevant} \), the number of in-neighbors and number of out-neighbors in \( Z_{\relevant} \) remains the same, but the size of \( Z_{\relevant} \) is reduced to \( \mathcal{O}(k) \). This, together with our irrelevant vertex rule (\Cref{redrule:bigtype}), gives the desired kernel. 
	
	The following is the second result of our paper.
	
	\begin{restatable}{theorem}{sfastheorem} \label{theo:fast}
		\sfast can be solved in  time $ 2^{\OO(\sqrt{k} \log k)} +   n^{\OO(1)} $.
	\end{restatable}
	
	
	We follow the approach of Alon, Lokshtanov, and Saurabh \cite{DBLP:conf/icalp/AlonLS09}, who developed a sub-exponential time {\sf FPT} algorithm for {\sc FAST}. However, due to the generality of our problem, we must  deviate from their method while doing the dynamic programming step. Our algorithm builds upon the chromatic coding technique introduced in \cite{DBLP:conf/icalp/AlonLS09}, integrating it with a dynamic programming approach applied to the reduced instance obtained through \cref{theo:fastkernel}.

	\section{Preliminaries}\label{sec:preli}

	Let $[n]$ represent the set of integers $\{1,\ldots, n\}$. For a directed graph $G$, we denote the vertex set by $V(G)$ and the arc set by $A(G)$.  For a subset of vertices $X \subseteq V(G)$, we use the notation $G-X$ to denote the graph $G'=G[V \setminus X]$. For a subset $F$ of arcs $F \subseteq A(G)$, we use the notation $G-F$ to denote the graph $G'=(V(G), A(G) \setminus F)$. Further for an ease of notation, we use $G -v$ (resp. $G -e$) when $X = \{v\}$ (resp. $F = \{e\}$) is singleton. Similarly for a subset $F \subseteq V(G) \times V(G)$, we use the notion $G + F$ to denote the graph $G'=(V(G), A(G) \cup F)$. For a directed graph $G$, a vertex set $X \subseteq V(G)$ forms a strongly connected component (SCC) in $G$ if $X$ is a maximal set of vertices such that for any pair of vertices $ u,v \in X$, there is a directed path from $u$ to $v$ and from $v$ to $u$ in the induced subgraph $G[X]$. For a cycle $C$ in a directed graph, we use $A(C)$ to denote set of the arcs in the cycle $C$.

	A directed graph is called a tournament if there is an arc between every pair of vertices. For a tournament $T$, an arc set $F \subseteq A(T)$ is called a feedback arc set of $T$ if there is no directed cycle in the graph $T-F$. For a directed graph $G$ and $S \subseteq V(G)$, we call a vertex $v \in V(G)$, a $S$-vertex if $v \in S$. A cycle (resp. path) in $G$ is called $S$-cycle (resp. $S$-path) if the cycle (resp. path) contains some $S$-vertex. A $S$-triangle is a $S$-cycle with exactly three vertices. For  a directed graph $D$, and vertex subset $S \subseteq V(D)$, the graph $D$ is said to be $S$-acyclic if $D$ has no $S$-cycle. For a tournament $T$ and $S \subseteq V(T)$, an arc set $F\subseteq A(T)$ is called a $S$-feedback arc set (in short, $ S $-fas) of  $T$ if there is no directed $S$-cycle in $T-F$, in other words, $T-F$ is $S$-acyclic. 
	
	Given a directed graph $G$ and a subset $F \subseteq A(G)$ of arcs, we denote $G \circledast F $ to be the directed graph obtained from $G$ by reversing all the edges of $F$. That is,
	if $\rev(F) =  \{ (u,v) |  (v,u) \in F\}$, then for $G \circledast F $ the vertex set is $V (G)$ and the arc set is  $A(G \circledast F) =  (A(G)  \cup \rev(F )) \setminus F$. For  ease of notation, when the set $F = \{e\}$ is singleton, we use $\rev(e)$ and  $G \circledast e$.


	For a vertex $v \in V(G)$, we denote $\arc(v)$ to be the set of arcs incident to the vertex $v$ in $G$. For a pair of disjoint vertex sets $P, Q \subseteq V(G) $, let $\arc(P, Q)$ denote the set of all the arcs between $P$ and $Q$ in $A(G)$, more formally $\arc(P, Q) = A(G) \cap \{(u,v)~|~ u \in P, v \in Q~\text{or}~u \in Q, v \in P \}$. Whenever we say that "every arc in $\arc(P,Q)$ is from $P$ to $Q$", we mean that there in no arc $(u,v) \in \arc(P,Q)$ with $u \in Q, v \in P$. 
	
	

	\subparagraph{Kernelization.} A parameterized problem $ \Pi $ is a subset of $ \Gamma^* \times \mathbb{N}$ for some finite alphabet $ \Gamma $. An instance of a parameterized problem consists of $(X, k)$, where $k$ is called the parameter. The notion of kernelization is formally defined as follows. A kernelization algorithm, or in short, a kernelization, for a parameterized problem $ \Pi \subseteq \Gamma^* \times \mathbb{N}$ is an algorithm that, given $(X, k) \in  \Gamma^* \times \mathbb{N}$, outputs in time polynomial in $|X| + k$ a pair $ (X', k')\in  \Gamma^* \times \mathbb{N} $ such that (a) $ (X, k) \in \Pi$ if and only if $(X', k') \in \Pi$ and  (b)  $|x'|,|k| \leq g(k)$, where $ g $ is some computable function depending only on $ k $. The output of kernelization $  (X', k') $ is referred to as the kernel and the function $ g $ is referred to as the size of the kernel. If $ g(k) \in k^{\OO(1)}  $ , then we say that $ \Pi $ admits a polynomial kernel.  We refer to the monographs \cite{DBLP:series/mcs/DowneyF99,DBLP:series/txtcs/FlumG06,DBLP:books/ox/Niedermeier06} for a detailed study  of the area of kernelization.
	
	Our algorithms make use of reduction rules which transform one instance of a
	problem to another instance of the same problem. We use $(T, S, k)$ to represent
	the instance given as input to each reduction rule, and $(T',S',k')$ to represent
	the (modified) instance output by the rule. We say that a reduction rule is
	{\em sound} if for every input instance $(T, S, k)$ the rule outputs an equivalent instance $(T,S',k')$. We update $T \rightarrow T', S \rightarrow S', k \rightarrow k'$
	to get the input instance $(T',S',k')$
	for further processing. We say that an instance $(T, S, k)$ is reduced with respect to a reduction rule  if none of the conditions of the rule is applicable    to $(T, S, k)$.

	\section{Characterization of \sfast as Arcs Reversal Problem}
	
	
	In this section, we demonstrate that the problem of deleting $k$ arcs to eliminate all directed $S$-cycles is equivalent to the problem of reversing $k$ arcs to achieve the same outcome. In particular our problem is equivalent to the following reversal problem. 
	
	
	\medskip
	
	\defparprob{{\sc Subset Feedback Arc Set Reversal }(\sfastr)}{A tournament $T=(V(T),A(T))$, a vertex set $S \subseteq V(T)$,  and  an    integer $ k$.}{$k$}{Find  $ F \subseteq A(T)$ with $|F| \leq k$ such that   $ T\circledast F $ has no $S$-cycle.}
	
	\medskip
	
	For our proof, we also need the following characterization relating $S$-cycle and $S$-triangles.  It is well known that a tournament is acyclic if and only if it does not contain any triangle~\cite{DBLP:journals/jda/DomGHNT10}. The following lemma, proved by Bai in \cite{DBLP:journals/tcs/BaiX23}  states a similar statement for the subset variant. 
	\begin{lemma}{\rm \cite[Lemma~2]{DBLP:journals/tcs/BaiX23}}
		\label{lem:Scycle}
		A tournament is $ S $-acyclic if and only if it does not contain an $S$-triangle. 
	\end{lemma}
	
	We present an observation regarding the presence of an $S$-triangle for every $S$-vertex that lies within some $S$-cycles. 
	
	\begin{observation}\label{obs:striangle}
		Let $s\in S$. Then, there is an $S$-cycle containing $s$ if and only if there is  an $S$-triangle containing $s$. 
	\end{observation}
	
	\begin{proof}
		We only prove the forward direction as the reverse direction holds vacuously. Let $C$ be a shortest $S$-cycle in a tournament $T$ that includes $s$. 
		Let $c_1 c_2 \ldots c_{\ell} c_1$ be the cycle where $c_1 = s$. If $\ell = 3$, then we are done. Otherwise, assume that $\ell \geq 4$. Consider the vertex $c_3$. If $(c_3,s) \in A(T)$, then  $s c_2 c_3 s$ is a triangle.  If $(s,c_3) \in A(T)$, then we get a shorter cycle $s c_3 c_4 \ldots s$ which is a contradiction to our assumption that $C$ is  a shortest $S$-cycle in a tournament $T$ that includes $s$. 
	\end{proof}
	
	Now we prove that  \sfast and \sfastr are equivalent. This is a  known result for {\sc FAST}.
	
	\begin{figure}[ht!]
		\begin{center}
			\includegraphics[scale=0.55]{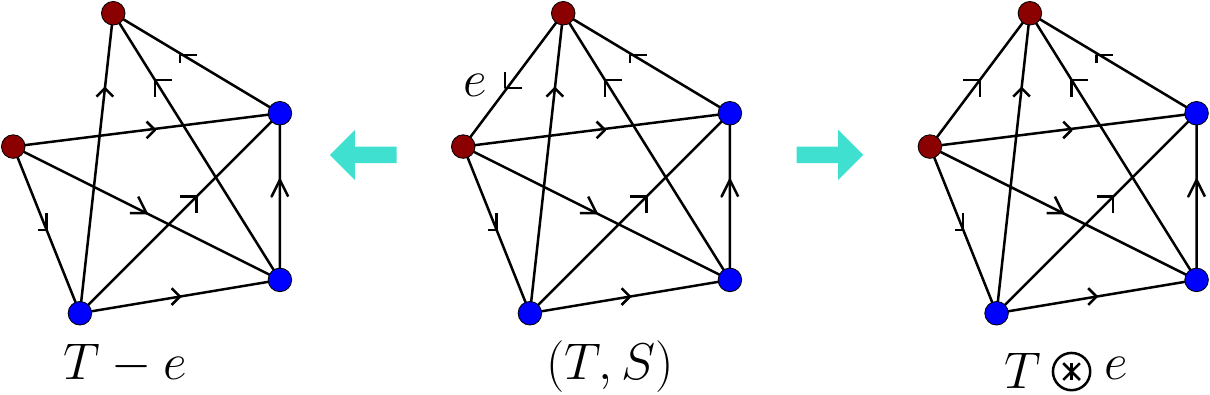}
		\end{center}
		\caption{Illustration of  \Cref{lem:minimal_equivalence}:  equivalence between \sfast  and \sfastr. Red vertices indicates $S$-vertices.}
		\label{fig:reversedeletion}
	\end{figure}
	
	\begin{lemma} \label{lem:minimal_equivalence}
		Let $T$ be a tournament, $S \subseteq V(T)$ and $ F \subseteq A(T)$. Then $F$ is an inclusion-wise minimal solution for \sfast if and only if $F$ is an inclusion-wise minimal solution for \sfastr. 
	\end{lemma}
	
	\begin{proof}

		In the forward direction, let $(T,S,k)$ be a \yes instance of \sfast and $F$ be a inclusion-wise minimal solution. We first claim that $F$ is also a solution for \sfastr. Assume towards contradiction that $T\circledast F$ has an $S$-cycle, say $C$. Let $f_1, f_2, \ldots, f_{\ell}$ be the arcs of $A(C) \cap \rev(F)$ in the order of their appearance on the cycle $C$ and $e_i = \rev(f_i) \in F$. Since $F$ is a inclusion-wise minimal solution of \sfast for $(T,S,k)$, this implies that for every $e_i \in F$, there exists an $S$-cycle $C_i$ in $T$ such that $ A(C_i) \cap F = \{e_i\}$.  Now consider a following closed walk $W$ in $T$. We start at some vertex of $C$ and follow the cycle $C$ except when an arc $f_i \in \rev(F)$ comes. When $f_i$ comes in the walk, then instead of following the arc $f_i$, we traverse the  path $C_i - e_i$.  By definition, $W$ is a closed walk in $T$ that does not contain any arc of $\rev(F)$ from which we can retrieve an $S$-cycle in $T - F$, which contradicts that $F$ is a solution for \sfast. $F$ is also an inclusion-wise minimal solution for \sfastr because if there is a proper subset $F' \subset F$  of arcs such that $F'$ is a solution for \sfastr, this will also be a solution for \sfast of smaller size which contradicts minimality of $F$.
		
		In the reverse direction, let $\widetilde{F}$ be an inclusion-wise minimal solution for \sfastr. It is easy to see that $\widetilde{F}$ is also a solution for \sfast. Moreover, $\widetilde{F}$ is also an inclusion-wise minimal solution for \sfast, Suppose not, consider a proper subset $\widetilde{F}' \subset \widetilde{F}$ which is a inclusion-wise minimal solution for \sfast, then by the argument of forward direction, this will become a inclusion-wise minimal solution for \sfastr, that contradicts the minimality of $\widetilde{F}$.
	\end{proof}
	
	
	By Lemma~\ref{lem:minimal_equivalence} we know that a minimal solution to \sfast is the same as a minimal solution to \sfastr. For an illustration of \Cref{lem:minimal_equivalence}, see \Cref{fig:reversedeletion}. Thus, from {\bf \em now onward we focus on the problem \sfastr}.

	\subsection{$\boldsymbol{S}$-Topological Ordering}
	Next we define $\spclorder$ which is useful for our kernel. This enriches the notion of canonical sequence defined by Bai and Xiao~\cite[Definition $1$] {DBLP:journals/tcs/BaiX23}. 
	
	\begin{definition}[$\spclorder$]\label{def:ordered_part}
		{\em   For a tournament $T,$ and $S \subseteq V(T)$, $T$ is said to have a $\spclorder$ $\Pi$ with respect to $S$ if $T$ has an ordered partition $V_1 \uplus V_2 \uplus \ldots \uplus V_{\ell}$ of $V(T)$ that satisfies the following:
			\begin{description}
				\item[{(i)}\label{p1}] For every $s \in S$, there exists $i \in [\ell]$ such that the part $V_i = \{s\}$ is singleton.  
				\item[{(ii)}\label{p2}] For every $i \in [\ell]$, the part $V_i$ forms a strongly connected component in $T$.
				\item[{(iii)}\label{p4}] For every pair $i < j \in [\ell]$, and every $u \in V_i$ and $v \in V_j$,  $(u,v)$ is an arc in $T$.
		\end{description}}
	\end{definition}
	
	We refer to a part consisting of a single $S$ vertex as a {\em $S$-singleton part}. Now we have the following lemma.

	\begin{lemma}\label{obs:ordering}
		Let $T$ be a tournament and $S \subseteq V(T)$. $T$ admits a $\spclorder$ $\Pi$ with respect to $S$ if and only if $T$ does not contain any $S$-triangle. Furthermore, if $T$ does not have an $S$-triangle, we can compute $\spclorder$  in polynomial time. 
	\end{lemma}
	
	\begin{proof}
		In the forward direction, let $\Pi$ be a $\spclorder$ of $T$ with respect to $S$. Let $V_1 \uplus V_2 \uplus \ldots \uplus V_{\ell}$ be the corresponding ordered partition of $V(T)$.  Assume towards contradiction that $T$ has an $S$-triangle with arcs $(a,s), (s,b)$ and $(b,a) \in A(T)$ where $s \in S$. Let $V_q$ be the $S$-singleton part that contains $s$. Since $(a,s) \in A(T)$, this implies that the vertices of the part $V_p$ that contain $a$, precede the vertices of $V_q$ in $\Pi$. Similarly, since $(s,b) \in A(T)$, this implies that the vertices of the part $V_q$ precede the vertices of the part $V_r$ that contains $b$ in $\Pi$. Hence, it follows that all the vertices of the part $V_p$ precede all the vertices of $V_r$ in $\Pi$ which is a contradiction that $(b,a) \in A(T)$ with $b \in V_r$ and $a \in V_p$.
		
		In the reverse direction, suppose that $T$ does not contain an $S$-triangle. We will now construct a $\spclorder$ $\Pi$ of $T$ as follows. First, compute the maximal strongly connected components $V_1, \ldots, V_\ell$ of $T$. Since these components are maximal, they are pairwise disjoint. These components will form the parts of $\spclorder$ $\Pi$. Observe that, since there are no $S$-triangles (or $S$-cycles), if $V_i$ contains a vertex $s\in S$, then $V_i=\{s\}$. This implies that for every $s \in S$, there exists $i \in [\ell]$ such that the component $V_i = \{s\}$ is a  singleton. We construct an auxiliary graph $T'$ where $V(T') = \{ v_1, \ldots,v_\ell\}$. Here, the vertex $v_i\in V(T')$ corresponds to  the strongly connected component $V_i$.  There is an arc from vertex $v_i$ to $v_j$ in $A(T')$  if and only if there is an arc from a vertex in $V_i$ to a vertex in $V_j$. Note that if an arc exists from a vertex in $V_i$ to a vertex in $V_j$, then every vertex in $V_i$ has an arc to every vertex in $V_j$; otherwise, $V_i\cup V_j$ would form a larger strongly connected component. Therefore, the arc set $A(T')$ is well-defined. Furthermore, we can observe that $T'$ is also a tournament. 
		
		Finally, note that $T'$ is acyclic; otherwise, the union of the maximal strongly connected components corresponding to a cycle in $T'$ would create a larger strongly connected component in $T$. Since $T'$ is an acyclic tournament, there is a unique topological ordering of it, which can be found in $\OO(m+n)=\OO(n^2)$ time~\cite{DBLP:books/daglib/0015106}. The corresponding order of the maximal strongly connected components is the $\spclorder$ $\Pi$ of $T$.  This concludes the proof. 
	\end{proof}

	\section{Bounding the number of Terminals} \label{subsec:bound_terminals}
	

	In this section, we design reduction rules to obtain an equivalent instance of \sfastr with at most $\OO(k)$ terminal vertices. That is, $|S|\leq \OO(k)$. For a gentle introduction, we first provide a slightly worse quadratic bound on the number of terminals, and then refine it to a linear bound.

	\subsection{Quadratic Upper Bound on the Number of Terminals}
	
	The first reduction rule handles sanity checks; its correctness is self-evident. 
	
	\begin{reduction rule}[Sanity Check]
		\label{redrule:sanity_check}
		Let $(T,S,k)$ be an instance of \sfastr where $k \leq 0$. If $k=0$ and $T$ has no $S$-triangle, return a trivial \yes instance; otherwise, return a trivial \no instance.
	\end{reduction rule}
	
	To reduce the number of terminal vertices, we further remove irrelevant $S$-vertices from the graph.  The following reduction rule is designed to accomplish this task.
	
	\begin{reduction rule}[Irrelevant Vertex Rule 1] \label{redrule:triangle_free_vertex}
		Let $(T,S,k)$ be an instance of \sfastr. If there is a vertex $s \in S$ such that $s$ does not participate in any triangle in $T$, then delete $s$ from $T$. The resulting instance is $(T'= T - s, S \setminus \{s\}, k).$
	\end{reduction rule}
	
	It is easy to verify that \cref{redrule:triangle_free_vertex} can be applied in polynomial time.  In the following, we prove the correctness of it.
	
	\begin{lemma}
		\cref{redrule:triangle_free_vertex} is sound.
	\end{lemma}
	
	\begin{proof}
		The key aspect of the proof is that the sets of $S$-cycles in $T$ and $T'$ are identical.
		
		In the forward direction, let $(T,S,k)$ be a \yes instance of \sfastr and $F \subseteq A(T)$ be a solution of size at most $k$. Clearly, $F'= F \setminus \arc(s)$ is a solution for $(T'= T - s, S \setminus \{s\}, k)$ of size at most $k$, where $\arc(s) \subseteq A(T)$ is the  set of arcs incident to the vertex $s$ in $T$.
		
		
		In the reverse direction, let $(T - s, S \setminus \{s\}, k)$ be a \yes instance of \sfastr and $\widetilde{F} \subseteq A(T)$ be a minimal solution of size at most $k$. We claim that, $\widetilde{F}$ is also a solution of size at most $k$ for $(T,S,k)$. Assume towards contradiction that there exists a $S$-triangle  in $T\circledast \widetilde{F}$. Since $\widetilde{F}$ was a solution for $(T', S \setminus \{s\},k)$, the triangle must contain the vertex $s$. Let $\triangle$ be such a triangle with the vertices $\{s, p, q\}$ and arcs $\{(s,p), (p,q), (q,s)\}$. Now if $(p,q) \notin \widetilde{F}$ we obtain a $S$-triangle in $T$ with vertices $\{s,p,q\}$, a contradiction to the fact that  $s$ does not participate in any triangle in $T$. Else $(q,p) \in  \widetilde{F}$. As $\widetilde{F}$ is minimal solution there is a $S$-cycle $C$ in $T-s$,  such that  $A(C) \cap \widetilde{F} = \{(q,p)\}$. That means there is a path from $p$ to $q$ in $T$ containing the arcs $A(C) \setminus (q,p)$. As $(s,p), (q,s) \in A(T)$, we obtain  a $S$-cycle $C_1$ in $T - \widetilde{F}$ where $s \in V(C_1)$. By \cref{obs:striangle},  $T$ has a $S$-triangle containing the vertex $s$, a contradiction.
	\end{proof}

	The next reduction rule is inspired by a similar rule for {\sc FAST}, which states that if an arc is contained in at least $k+1$ triangles, we must reverse this arc and reduce the budget by one~\cite[Theorem $2.8$]{DBLP:books/sp/CyganFKLMPPS15}. 
	
	\begin{reduction rule}[Many $S$-triangles Rule] 
		\label{redrule:many_triangle}
		Let $(T,S,k)$ be an instance of \sfastr. Let $e \in A(T)$ be an arc in $T$ such that $e$ participates in $k+1$ $S$-triangles, then reverse the arc $e$. The resulting instance is $(T \circledast e, S, k-1)$ 
	\end{reduction rule}
	See Figure~\ref{fig:manytriangle} for an illustration of Reduction Rule~\ref{redrule:many_triangle}. 
	\begin{figure}[ht!]
		\begin{center}
			\includegraphics[scale=0.55]{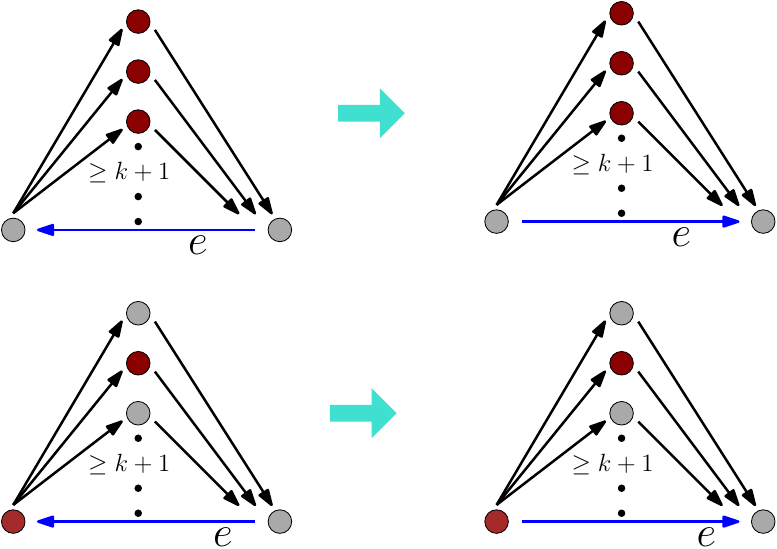}
		\end{center}
		\caption{Illustration of  \Cref{redrule:many_triangle} (Many $S$-triangle). Vertices with color red are from   $S$.}
		\label{fig:manytriangle}
	\end{figure}
	
	It is easy to verify that \cref{redrule:many_triangle} can be applied in polynomial time.  In the following, we prove the correctness of it.
	
	\begin{lemma}
		\cref{redrule:many_triangle} is sound.
	\end{lemma}
	
	\begin{proof}
		In the forward direction, let $(T,S,k)$ be a \yes instance of \sfastr and $F \subseteq A(T)$ be a solution of size at most $k$. So $T\circledast F$ does not contain any $S$-triangle. Observe that $F$ must include the arc $e$. If it does not, then to hit all the $(k+1)$ $S$-triangles that contain the arc $e$, we need to reverse at least $k+1$ arcs, which is not feasible since we are allowed to reverse at most $k$ arcs. We claim that $F'= F \setminus \{e\}$ is a solution of size at most $k-1$ for $(T'=T \circledast e, S, k-1)$.  Notice that $T\circledast F$ and $T'\circledast F'$ are the same tournaments. So, if $T'\circledast F'$ has a $S$-triangle, then that will imply that $T\circledast F$ also has a $S$-triangle which contradicts that $F$ is a solution for $(T,S,k)$. 
		
		In the reverse direction, let $(T'=T \circledast e, S, k-1)$ be a \yes instance of \sfastr and  $\widetilde{F}$ be a solution of size at most $k-1$. Consider $\widetilde{F}'= \widetilde{F} \cup \{e\}$. We claim that $\widetilde{F}'$ is a solution of size at most $k$ for $(T,S,k)$. 
		Now note that $T'\circledast \widetilde{F}$ and $T\circledast \widetilde{F}'$ are the same tournaments. So, if $T\circledast \widetilde{F}'$ has a $S$-triangle, then that will imply that $T'\circledast \widetilde{F}$ also has a $S$-triangle which contradicts that $\widetilde{F}$ is a solution for $(T',S,k-1)$. 
	\end{proof}
	
	\begin{reduction rule}[Bounded Terminal Rule] \label{redrule:no_instance}
		Let $(T,S,k)$ be an instance where none of the  Reduction Rules \ref{redrule:sanity_check}-\ref{redrule:many_triangle} are applicable. If $|S| \geq (k+1)^2$, then return a trivial \no instance.
	\end{reduction rule}
	
	It is easy to verify that \cref{redrule:no_instance} can be applied in polynomial time.  In the following, we prove the correctness of it.
	
	\begin{lemma}
		\cref{redrule:no_instance} is sound.
	\end{lemma}
	
	\begin{proof}
		Let $(T,S,k)$ be a \yes instance for \sfastr where none of the Reduction Rules \ref{redrule:sanity_check}-\ref{redrule:many_triangle} are applicable. Let $F \subseteq A(T)$ be a solution of size at most $k$. Let $V(F)$ be the set of vertices that are incident to at least one of the arcs of $F$. Clearly $|V(F)| \leq 2k$. Now, consider $V(T) \setminus V(F)$. Since \cref{redrule:triangle_free_vertex} is not applicable to $(T,S,k)$, every $S$-vertex in $V(T) \setminus V(F)$  participates in some $S$-triangle. Moreover, since $F$ is a solution, one of the arcs of that $S$-triangle must be in $F$. Since \cref{redrule:many_triangle} is not applicable to $(T,S,k)$, it follows that for every arc $e \in F$, there are at most $k$ $S$-vertices that are involved in $S$-triangles where $e$ serves as a common arc and the other arcs are disjoint. So, the number of such $S$-vertices in $V(T) \setminus V(F)$ is upper bounded by $k^2$. So, the overall number of $S$-vertices is bounded by $k^2 + 2k =  (k+1)^2 -1$. 
	\end{proof}
	
	Let $(T,S,k)$ be an instance where none of the Reduction Rules \ref{redrule:sanity_check}-\ref{redrule:no_instance} are applicable. This immediately bounds the number of $S$-vertices in $(T,S,k)$. We formalize it in the following Lemma.
	
	\begin{lemma}\label{lem:ter}
		Let $(T,S,k)$ be an instance where none of the Reduction Rules \ref{redrule:sanity_check}-\ref{redrule:no_instance} are applicable. Then the number of terminals, that is, $|S|$ is upper bounded by $(k+1)^2$.
	\end{lemma}
	
	
	\subsection{Linear Upper Bound on the Number of Terminals}
	In this section, we aim to improve the bound of $|S|$ from quadratic to linear, i.e., $\OO(k^2)$ to $\OO(k)$. Our procedure for establishing a linear upper bound on the number of terminals is inspired by the $4k+1$ vertex kernel for {\sc FAST} obtained by Paul, Perez and Thomassé~\cite{DBLP:journals/jcss/PaulPT16}, as well as the earlier linear kernel of size $(2+\epsilon)k$ developed by Bessy et. al.~\cite{DBLP:journals/jcss/BessyFGPPST11}. For consistency with the literature, we adopt the notation used in~\cite{DBLP:journals/jcss/PaulPT16}. 
	
	Consider an instance $(T,S,k)$ of \sfastr such that none of the Reduction Rules \ref{redrule:sanity_check} - \ref{redrule:no_instance} are applicable to the instance. By \cref{lem:ter}, we know that $|S|$ is upper bounded by $(k+1)^2$. If $ |S| \leq 4k $, then we are already done.  So from now onwards we assume that $ |S| \geq 4k+1 $. 
	
	We use $ T_{\sigma} = (V , A, \sigma) $ to denote a tournament whose vertices are arranged according to a fixed ordering $ \sigma = v_1, \ldots, v_n$ on $ V(T) $. We say that an arc $ (v_i, v_j) $ of $ T_{\sigma} $ is a {\em backward} arc if $ i >j $, otherwise {\em forward}. We say that an arc $ (v_j, v_i) $ of $ T_{\sigma} $ is a {\em $S$-backward arc} if $ j >i $ and either of (i), (ii) holds: (i)  one of $v_i, v_j$ is a $S$-vertex, (ii) there exists a $q$  such that $i < q < j$ such that $v_{q} \in S$.  
	
	A partition $ P = \{V_1, \ldots , V_\ell\} $ of $ V $ is an ordered partition of an ordered tournament $T_{\sigma} $ if for every $ i \in [\ell] $, $ V_i $ is a set of consecutive vertices (interval) in $ \sigma $. Hereafter, an ordered partition of $T_{\sigma} $ will be denoted by $ P_{\sigma} = (V_1, \ldots, V_{\ell}) $, with the assumption that for every $ x \in V_i $ and $ y \in V_j $ where $ 1 \leq  i <j \leq  \ell $, it holds that $x <_{\sigma} y$. We call  all arcs between the intervals (having their
	endpoints in different intervals) as external arcs of $ P_{\sigma} $ and denote them by   $ A_B(T_{\sigma}, P_{\sigma})$, i.e., $ A_B(T_{\sigma}, P_{\sigma}) = \{(u,v)~|~u \in V_i, v \in V_j,i \neq j\} $. Let $ A_I(T_{\sigma}, P_{\sigma}) = A(T) \setminus A_B(T_{\sigma}, P_{\sigma})$. See \Cref{fig:partition3} for an illustration of an ordered partition $P_{\sigma}$. 
	
	\begin{figure}[ht!]
		\begin{center}
			\includegraphics[scale=0.5]{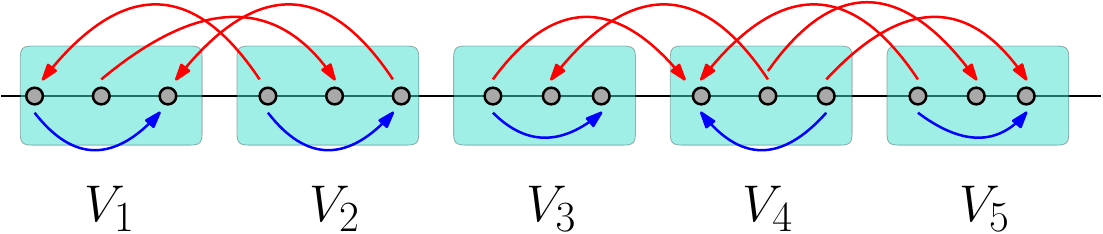}
		\end{center}
		\caption{Illustration of an ordered partition $ P_{\sigma} = (V_1, \ldots, V_{5}) $. Red and blue arcs signifies the arc sets $A_B(T_{\sigma}, P_{\sigma})$ and $A_I(T_{\sigma}, P_{\sigma})$, respectively.}
		\label{fig:partition3}
	\end{figure}
	
	\begin{definition}[certificate of an arc]	
		{\em Consider an ordered partition $ P_{\sigma} = (V_1, \ldots, V_{\ell}) $ of an ordered tournament  $T_{\sigma}$ and let $ f = (v,u) $ be a backward arc of $T_{\sigma} $. We refer to a {\em certificate} of $f$, denoted by $ c ( f ) $, as 
			any directed    $S$-path from $ u $ to $ v $ that utilizes only forward arcs from $ A_B(T_{\sigma}, P_{\sigma}) $.}
	\end{definition}
	
	\begin{definition}[certification of an arc set]	
		{\em Consider an ordered partition $ P_{\sigma} = (V_1, \ldots, V_{\ell}) $ of an ordered tournament  $T_{\sigma}$  and let $ F $ be a set of some backward arcs of $T_{\sigma} $. 
			We say that	{\em we can certify} $ F $ whenever it is possible to find a set $ \FF= \{c(f) : f \in F\} $ of arc-disjoint  certificates for the arcs in $ F $.}
	\end{definition}

	\begin{definition}[safe partition]
		
		{\em Consider an ordered partition $ P_{\sigma} = (V_1, \ldots, V_{\ell}) $ of an ordered tournament  $T_{\sigma}$. Let $B_E$ denote the set of all backward arcs in $ A_B(T_{\sigma}, P_{\sigma}) $. For a vertex subset $S \subseteq V(T)$, we say $ P_{\sigma}$ is  a
			{\em $ S $-safe partition}  if   we can certify $B_E$. 
				
		}
		
	\end{definition}
	
	Let $D_{\sigma} = (V_{\sigma} , A)$ be an ordered directed graph. We use $\textsf{Sfas}(D_{\sigma},S) $  to denote the size of a minimum subset feedback arc set in the directed graph $D_{\sigma}$. Specifically, given a directed graph  $D_{\sigma}$, and a vertex subset $S \subseteq V_{\sigma}$, it refers to the cardinality of the smallest
	set $ F \subseteq A$ of arcs whose removal (or reversal) renders $ D_{\sigma} $ $ S $-acyclic. Given a tournament, for an arc subset $A' \subseteq A[T]$, the directed graph $T[A']$ is defined by the subgraph $T$ containing the arc set $A'$. More specifically, $T[A']= (V(T), A')$.  For a ease of notation, when the subset is clear to the context, we use $\textsf{Sfas}(D_{\sigma}) $ instead of $\textsf{Sfas}(D_{\sigma}) $.
	
	\begin{lemma}\label{lem:safe}
		Consider an ordered partition $ P_{\sigma} = (V_1, \ldots, V_{\ell}) $ of an ordered tournament  $T_{\sigma}$ and a terminal set $S \subseteq V(T)$. If $ P_{\sigma}$ is  a $ S $-safe partition then, 
		
		$$ \textsf{Sfas}(T_{\sigma}) = \textsf{Sfas}(T_{\sigma} [A_I(T_{\sigma}, P_{\sigma})])+\textsf{Sfas}(T_{\sigma} [A_B(T_{\sigma}, P_{\sigma})])$$
		Moreover, there
		exists a minimum sized $S$-feedback arc set of $ T_{\sigma} $ containing $ B_E $, where $B_E$ denote the set of all backward arcs in $ A_B(T_{\sigma}, P_{\sigma}) $.
		
	\end{lemma}
	
	\begin{proof}
		Given any bipartition of the arc set \( A \) into \( A_1 \) and \( A_2 \), we have 
		$$
		\textsf{Sfas}(T_{\sigma}) \geq \textsf{Sfas}(T_{\sigma}[A_1]) + \textsf{Sfas}(T_{\sigma}[A_2]).$$
		Specifically, for the partition of \( A \) into \( A_I(T_{\sigma}, P_{\sigma}) \) and \( A_B(T_{\sigma}, P_{\sigma}) \), it follows that 
		$$\textsf{Sfas}(T_{\sigma}) \geq \textsf{Sfas}\big(T_{\sigma}[A_I(T_{\sigma}, P_{\sigma})]\big) + \textsf{Sfas}\big(T_{\sigma}[A_B(T_{\sigma}, P_{\sigma})]\big).$$ 
		Next, we need to show that 
		$$
		\textsf{Sfas}(T_{\sigma}) \leq \textsf{Sfas}\big(T_{\sigma}[A_I(T_{\sigma}, P_{\sigma})]\big) + \textsf{Sfas}\big(T_{\sigma}[A_B(T_{\sigma}, P_{\sigma})]\big).
		$$
		
		This assertion holds because, after reversing all arcs in \( B_E \), each remaining directed \( S \)-cycle is contained within \( T_{\sigma}[V_i] \) for some \( i \in [\ell] \). In other words, once all arcs in \( B_E \) are reversed, every \( S \)-cycle is entirely contained within \( T_{\sigma}[A_I(T_{\sigma}, P_{\sigma})] \). Observe that as $ P_{\sigma}$ is   $ S $-safe,    the set of all backward arcs of  $A_B(T_{\sigma}, P_{\sigma})$, i.e., \( B_E \) can be certified using only arcs from $A_B(T_{\sigma}, P_{\sigma})$. This concludes the proof of the first part of the lemma. Essentially, we have shown the existence of a minimum-sized \( S \)-feedback arc set for \( T_{\sigma} \) that includes \( B_E \).  This completes the proof of the lemma.
	\end{proof}
	
	\subparagraph{Construction of  a $\boldsymbol{S}$-safe partition.}
	Recall that an \( S \)-triangle is a directed cycle of length three that includes at least one vertex from \( S \). It is clear that the number of arc-disjoint \( S \)-triangles provides a lower bound for the size of the smallest \( S \)-feedback arc set in a tournament. We illustrate how this set can be utilized to identify a safe partition in polynomial time.
	
	
	\begin{definition}[Conflict $S$-packing]
		{\em  A {\em conflict $S$-packing} of a tournament $ T $ is a maximal collection of arc-disjoint $S$-triangles.}
	\end{definition}
	Observe that a conflict \( S \)-packing can be computed in polynomial time, for instance, by using the following \emph{greedy} algorithm: find an \( S \)-triangle, delete its arcs, and continue. Given an \( S \)-conflict packing \( \mathcal{C} \) of an instance \( (T, S, k) \) of \sfastr, let \( V(\mathcal{C}) \) denote the set of vertices covered by \( \mathcal{C} \), and \( A(\mathcal{C}) \) denote the set of all arcs contained in \( \mathcal{C} \). See \Cref{fig:partition2} for an illustration of conflict \( S \)-packing. The following observation is immediate.
	
	\begin{observation}\label{obs:cfno}
		Let $ \CC $ be a conflict $ S $-packing of a tournament $ T $. If $ (T,S,k) $ is an \yes instance of  \sfastr, then $ |\CC| \leq   k $ and $ |V (\CC)| \leq   3k $.
		
	\end{observation}
	
	\begin{figure}[ht!]
		\begin{center}
			\includegraphics[scale=0.5]{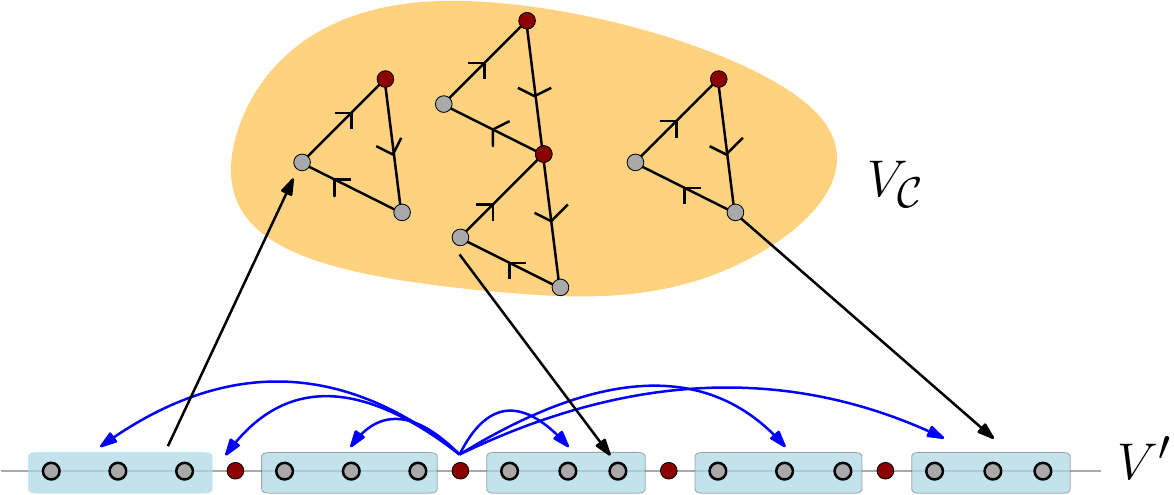}
		\end{center}
		\caption{Illustration of Conflict $S$-packing. Vertices with color red are $S$-vertices.}
		\label{fig:partition2}
	\end{figure}
	
	\begin{lemma}[Conflict $S$-packing Lemma]\label{lem:conpack}
		
		Let $ \CC $ be a conflict $ S $-packing of a tournament $ T $ and let $S_1 \coloneq S \setminus V(\mathcal{C})$. There exists a vertex ordering $ \sigma $ of $ T $ such that
		every $ S_1 $-backward arc $ (u,v) $ of $ T_{\sigma} $ satisfies $ (u,v) \in A(\CC) $.
	\end{lemma}
	
	\begin{proof}
		Consider the vertex set \( \{v_1, v_2, \ldots, v_n\} \) in \( T \). Let \( \mathcal{C} \) be a conflict \( S \)-packing of a tournament \( T \). We aim to construct a unique vertex ordering \( \sigma \) of \( T \). Define the vertex set \( S_1 \coloneq S \setminus V(\mathcal{C}) \). Since \( |S| \geq 4k + 1 \), it follows that \( S_1 \neq \emptyset \). Let \( V' = V \setminus V(\mathcal{C}) \).
		
		By the maximality of \( \mathcal{C} \), the induced subgraph \( T[V'] \) is \( S \)-acyclic. Therefore, there exists a unique topological ordering \( \sigma_{S_1} \) for the vertices of \( S_1 \). For all other vertices in \( V' \), specifically for each non-\( S \)-vertex \( u \) in \( V' \), we know that \( T[S_1 \cup \{u\}] \) is acyclic. Thus, for each non-\( S \)-vertex \( u \in V' \), there is a unique pair \( p, q \) of consecutive \( S_1 \) vertices in \( \sigma_{S_1} \) such that inserting \( u \) between \( p \) and \( q \) yields the topological ordering of \( T[S_1 \cup \{u\}] \).  In this way we obtain an ordering \( \sigma_{V'} \) of \(V'\).

		Our next aim is to insert the vertices of \( V(\mathcal{C}) \) into the ordering  \( \sigma_{V'} \). Observe that for each \( v \in V(\mathcal{C}) \), \( T[S_1 \cup \{v\}] \) is acyclic; in fact, \( T[V' \cup \{v\}] \) is \( S \)-acyclic. Therefore, for every vertex \( v \in V(\mathcal{C}) \), there exists a unique pair \( p, q \) of consecutive \( S_1 \) vertices in \( \sigma_{S_1} \) such that inserting \( v \) between \( p \) and \( q \) yields the topological ordering of \( T[S_1 \cup \{v\}] \). 
		

		However, in this manner, it is possible to insert more than one non \( S \)-vertex from \( V' \) as well as  vertices from \( V(\mathcal{C}) \) between \( p \) and \( q \). Let \( {\sf Bag}_{pq} \) denote the set of non-\( S \) vertices from \( V' \) and vertices from \( V(\mathcal{C}) \) that are inserted between \( p \) and \( q \). 
		To ensure the formation of a unique vertex ordering, we adopt the following convention: for any pair of  vertices \( v_i, v_j \) in \( {\sf Bag}_{pq} \), we define \( v_i <_{\sigma} v_j \) if and only if \( i < j \). This approach allows us to obtain a unique topological ordering \( \sigma \) for $T$.

		It remains to show that every $ S_1 $-backward arc $ (u,v) $ of $ T_{\sigma} $ satisfies $ (u,v) \in A(\CC) $. Assume towards contradiction that there is a $ S_1 $-backward arc $ (u,v) $ of $ T_{\sigma} $ such that $ (u,v) \notin A(\CC) $. Observe that according to our construction of $\sigma$, no $S_1$-backward arc is incident to a vertex in $S_1$. This indeed imply that   there is a $S_1$-vertex $w$ such that $v <_{\sigma} w <_{\sigma} u$. Moreover, $\{(v,w), (w,u)\} \subseteq  A(T)$. Now $w \notin V(\CC)$ and $ (u,v) \notin A(\CC) $ together imply that   none of the arcs in $\{(u,v), (v,w), (w,u)\}$ are contained in $A(\CC)$, a contradiction to the maximality of $\CC$.  \end{proof}
		

Let $ \CC $ be a conflict $S$-packing of $T$ and $\sigma$ be the vertex ordering obtained by \cref{lem:conpack}. 

\medskip 

\noindent Consider the auxiliary bipartite graph $H= (X \uplus Y, E)$ being constructed as follows:
\begin{itemize}
	\item For every arc $ (u,v) $ in $\CC$, we add a vertex $x_{u,v}$ in $ X $,
	\item $ Y= S \setminus V(\CC) $,
	\item For a pair of vertices $x_{u,v} \in X$ and $y \in Y$ we add an edge $ x_{u,v}y$ in $ E(H) $ if and only if   $\{u,v,y\}$ forms a directed triangle  in $ T $, i.e., $(y,u), (v,y) \in A(T)$. 
\end{itemize}

\begin{observation}
	\label{obs:isolated}
	No vertex of $ Y $ is isolated in $ H $.
\end{observation}

\begin{proof}
	Assume towards contradiction that the graph $H$ has an isolated vertex $y \in Y$. This implies that $y$ does not belong to any triangle that contains an arc from $A(\CC)$. 
	But since $y$ is a $S$-vertex and our current instance is reduced with respect to  \cref{redrule:triangle_free_vertex}, the vertex $y$ must belongs to some triangle, say $\triangle$.  Now $A(\CC)$ contains no arc of $\triangle$, a contradiction to the maximality of $\CC$.
\end{proof}

The following fact is known on matching and vertex cover in bipartite graphs.

\begin{proposition}[\cite{DBLP:journals/jcss/PaulPT16}, Lemma 2.1.]\label{prop:vc}
	Let $ B= (V_1 \uplus V_2, E)$ be a  bipartite graph, $ M $ be a minimum vertex cover of $ B $, $ M_1 = V_1 \cap M, M_2 = V_2 \cap M $. Any subset $I \subseteq M_1$ can be matched into $V_2 \setminus M_2$.
\end{proposition}

Observe that $ |X| \leq 3k $ and $ |Y| > k+1 $ (since $ |S| \geq 4k+1 $). It is easy to follow  that any pair of edges in $H$ with no common end-point corresponds to a pair of arc-disjoint $S$-triangle. This indeed imply, any matching in $ H $ of size at least $ k+1 $ corresponds to a conflict $ S $-packing of size at least $ k+1 $, in which case the given instance $ (T,S,k) $ is a \no instance (\cref{obs:cfno}). So the maximum matching has size at most $ k $, then by  König-Egervary’s theorem \cite{DBLP:books/others/BondyM76}, a minimum vertex cover $ M $ of $ H $ has size at most $ k $. We denote $ M_X = M \cap X $ and $ M_Y= M \cap Y $. As $ |Y| > k+1 $ and $|M_Y| \leq k$, we have  that  $ Y \setminus M_Y \neq \emptyset$. 

Let $\sigma$ be the vertex ordering obtained by \cref{lem:conpack}. Let $ P_{\sigma} $ be the ordered partition of $ T_{\sigma} $ such that every part $ V_i $ consists of either a $S$-vertex corresponding to a vertex of $ Y \setminus M_Y $ or a maximal subset of consecutive vertices   in $ \sigma $ without containing  any vertex from the vertex set  corresponding to  $ Y \setminus M_Y $. See \Cref{fig:orderparition} for an illustration of the construction of the ordered partition $ P_{\sigma} $ based on the bipartite graph $H$.

\begin{figure}[t!]
	\begin{center}
		\includegraphics[scale=0.55]{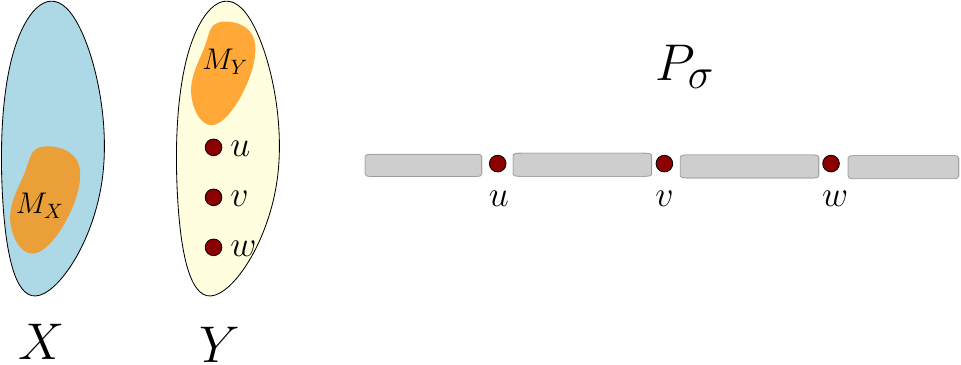}
	\end{center}
	\caption{Illustration of the construction of the ordered partition $ P_{\sigma} $ based on the bipartite graph $H= (X \uplus Y, E)$ The set $M_X \cup M_Y$ is a minimum sized vertex cover in $H$.}
	\label{fig:orderparition}
\end{figure}

\begin{claim}
	$ P_{\sigma} $ is a $S$-safe partition of $ T_{\sigma} $.
\end{claim}

\begin{proof}
	{Let $B_E$ denote the set of all backward arcs in $ A_B(T_{\sigma}, P_{\sigma}) $.
		Toward proving the claim we need to show that we can certify $B_E$. Consider an arc $e= (u,v)$ in $B_E$. Since $e$ is a $(Y \setminus M_Y)$-backward, $e$ is indeed a $S_1$-backward arc. Now by \cref{lem:conpack}, $(u,v) \in A(\CC)$, there is a vertex $x_{u,v}$ in $X$ (as per the construction of $H$). So each  arc in $B_E$ corresponds to a vertex in $X$. As $(u,v)$ is $(Y \setminus M_Y)$-backward, there exists a 
		a vertex $w \in (Y \setminus M_Y)$ such that $\{u,v,w\}$ forms a  triangle, moreover $S$-triangle in $T$. So $x_{uv}w \in E(H)$.  Since $M$ is a vertex cover of $H$ and $w \notin M$ that imply $x_{u,v} \in M$. Thereby there is a subset $X'' \subseteq X$ corresponding to the arcs in $B_E$ such that $X'' \subseteq M_X$. By \cref{prop:vc}, we know that $X''$ can be matched into $Y \setminus M_Y$ in $H$. Since every vertex of $Y \setminus M_Y$ is singleton in $P_{\sigma}$, the existence of the matching shows that the arcs in $B_E$ can be certified in $P_{\sigma}$. More specifically, if $x_{u,v}w$ is a matching edge then the certificate of $(u,v)$ is the $S$-path from $v$ to $u$ consisting the arc $(v,w)$ and $(w,u)$.}
\end{proof}

Note that a maximum matching $M$ can be found in polynomial time and the auxiliary bipartite graph $H$ can also be constructed in polynomial time.
Moreover, if \( |S| > 4k\), then $Y = S \setminus M_Y$ contains a vertex $s$ and by \cref{obs:isolated}, $s$ is not isolated in $H$.
This implies that there is a triangle $s,p,q$ with arcs $(s,p),(p,q),(q,s)$ where $q$ comes before $s$ and $p$ comes after $s$ in the computed ordering.
Thus, $(p,q)$ is an $Y$-backward arc and this implies the following observation.

\begin{observation}\label{obs:safe}
	If \( |S| > 4k \), then we can find a non-empty subset \( S_1 \subseteq S \) and a \( S_1 \)-safe partition $ P_{\sigma} $ with $B_E \neq \emptyset$ in polynomial time. 
\end{observation}

\begin{reduction rule}[Safe partition  Rule] 
	\label{redrule:safepartition}
	Consider a safe  partition $ P_{\sigma} = (V_1, \ldots, V_{\ell}) $ of an ordered tournament  $T_{\sigma}$. Let $B_E$ denote the set of all backward arcs in $ A_B(T_{\sigma}, P_{\sigma})$. If   $B_E \neq \emptyset$, then reverse all the arcs of $B_E$ and decrease $k$ by $|B_E|$, in other words, return the instance $(T \circledast B_E, S, k- |B_E|)$.
\end{reduction rule}

In the following, we prove the correctness of \cref{redrule:safepartition}.

\begin{lemma}
	\cref{redrule:safepartition} is sound and can be applied in polynomial time.
\end{lemma}

\begin{proof}
	Let $P_{\sigma}$ be a safe partition of $T_{\sigma}$. Observe that it is possible to certify all the backward arcs $B_E$, in $P_{\sigma}$. Hence, using \cref{lem:safe}, $T$ has a set of $S$-feedback arc set  of size at most $k$
	if and only if the tournament $T \circledast B_E$ has a $S$-feedback arc set of size at most
	$k - | B_E |$.
	
	The time complexity follows from the fact that we can compute a safe partition in polynomial time (by \cref{obs:safe}).
\end{proof}

From now onwards, we consider instances where Reduction Rules \ref{redrule:sanity_check} - \ref{redrule:safepartition} are not applicable in the corresponding tournament. So we can safely assume that  the number of terminals, that is, $|S|$ is upper bounded by $4k$. We formalize it in the following Lemma.

\begin{lemma}\label{lem:ter2}
	Let $(T,S,k)$ be an instance where none of the Reduction Rules \ref{redrule:sanity_check}-\ref{redrule:safepartition} are applicable. Then the number of terminals, that is, $|S|$ is upper bounded by $4k$.
\end{lemma}

\section{Irrelevant Vertex for  \sfast} \label{sec:sfaskernel}
In this section, we introduce reduction rules to eliminate non-terminal vertices. The idea is as follows: Let \( (T, S, k) \) be a \yes instance of \sfastr, and let \( F \) be a solution of size at most \( k \). We aim to understand the structure of a special ordering \( \Pi \) of \( T \circledast F \). Specifically, we focus on the interval between two consecutive vertices of \( S \) in \( \Pi \). If the number of vertices in this interval is large, we demonstrate that there exists an irrelevant vertex that can be safely deleted without affecting the solution. Notably, for any vertex \( s \in S \), all but at most \( k \) of its in-neighbors must appear before \( s \) in \( \Pi \), and all but at most \( k \) of its out-neighbors must appear after \( s \) in \( \Pi \). This observation is crucial for identifying the majority of vertices in the interval between two consecutive vertices of \( S \) in \( \Pi \). We formalize this idea using an equivalence relation.

\subsection{Partitioning Non Terminal  Vertices into Equivalence Classes} \label{subsec:partition_nont_eqvcls}
Let $(T,S,k)$ be an instance of \sfastr where none of the Reduction Rules \ref{redrule:sanity_check} - \ref{redrule:safepartition} are applicable. By \Cref{lem:ter2}, $|S| \leq 4k$.



Observe that each non-\( S \) vertex \( v \in V(T) \setminus S \) divides the \( S \) vertices into two groups: one group contains all of its in-neighbors from \( S \), and the other contains all of its out-neighbors from \( S \) in \( T \). To encapsulate this phenomenon, we define a function \( \type_T^S \colon V(T) \setminus S \to 2^S \) as follows: for every vertex \( v \in V(T) \setminus S \),
$$
\type_T^S(v) \coloneqq \{ s \in S \mid (s, v) \in A(T) \}.
$$
We can then categorize all non-\( S \) vertices into equivalence classes, where each class consists of vertices of the same type. More precisely, for a subset $X \subseteq S$, let \( \eqv_T^S[X] \) denote the set of all non-\( S \) vertices \( v \) for which \( \type_T^S(v) = X \); that is,
$$ \eqv_T^S[X] \coloneqq \{ v \in V(T) \setminus S \mid \type_T^S(v) = X \}.$$ 
We refer to \( \eqv_T^S[X] \) as \textit{non-trivial} if \( \eqv_T^S[X] \neq \emptyset \). It is noteworthy that there can be at most \( 2^{|S|} \) equivalence classes. In fact, we will show that the number of equivalence classes can be upper-bounded by $5k+1$.

\begin{reduction rule}\label{redrule:many_types}
	Let $(T,S,k)$ be an instance where none of the Reduction Rules \ref{redrule:sanity_check} - \ref{redrule:safepartition} are applicable. If the number of non-trivial equivalence classes in $T$ is more than $5k+1$, then return a trivial \no instance. 
\end{reduction rule}

It is easy to verify that \cref{redrule:many_types} can be applied in polynomial time.  In the following, we prove the correctness of it.

\begin{lemma}
	\cref{redrule:many_types} is sound.
\end{lemma}

\begin{proof}
	Let \((T,S,k)\) be a \yes instance for \sfastr, and let \(F \subseteq A(T)\) be a minimal solution of size at most \(k\). It follows that \(T \circledast F\) has a \(\spclorder\) \(\Pi\) (as shown in \cref{lem:minimal_equivalence}). In this \(\spclorder\) \(\Pi\), each vertex can position itself in one of three ways: between two \(S\)-vertices, before all \(S\)-vertices, or after all \(S\)-vertices. Additionally, all vertices within the same equivalence class will occupy identical positions relative to the \(S\)-vertices. Since there are a total of \( |S| + 1\) such positions concerning the \(S\)-vertices, it follows that \(T \circledast F\) has at most \( |S| + 1\) non-trivial equivalence classes.
	
	Next, we aim to show that \(T\) contains no more than \( |S| + k + 1\) non-trivial equivalence classes. Notice that a non-\(S\) vertex \(u \in V(T) \setminus S\) will change its \(\type_T^S(u)\) if an arc \(e\) incident to \(u\) is included in \(F\). Furthermore, the other endpoint of \(e\) must be some vertex in \(S\) since \(\type_T^S(u)\) is determined exclusively by the \(S\)-vertices. Given that the solution can encompass up to \(k\) arcs, this means that at most \(k\) non-\(S\) vertices can have their respective types altered following the reversal of \(F\). Consequently, at most \(k\) non-trivial (non-empty) equivalence classes can become trivial after the reversal operation. Therefore, prior to the reversal, there were at most \( |S| + 1 + k\) non-trivial equivalence classes in \(T\). 
\end{proof}


So we have the following Lemma.

\begin{lemma}\label{lem:ter3}
	Let $(T,S,k)$ be an instance where none of the Reduction Rules \ref{redrule:sanity_check}-\ref{redrule:many_types} are applicable. Then the number of non-trivial equivalence classes in $T$ is bounded by  $5k+1$.
\end{lemma}



\subsection{Quadratic Bound on Relevant Vertices in Each Equivalence Class} \label{subsec:quadratic_bound} 
Let \((T,S,k)\) be an instance of \sfastr where none of the Reduction Rules \ref{redrule:sanity_check} - \ref{redrule:many_types} are applicable. This indicates that the number of non-trivial equivalence classes in \(T\) is bounded above by \(5k+1\). In this section, we provide an upper limit on the size of each equivalence class. If the size of each equivalence class is limited to \(7k+4\), we can conclude our analysis. Otherwise, we consider a non-trivial equivalence class \(\eqv_T^S[X]\) with a size exceeding \(7k+4\). For clarity, let’s rename \(\eqv_T^S[X]\) as \(Z\). Recall that \(Z\) corresponds to a partition \(S_1(=X) \uplus S_2\) of \(S\), such that for every vertex \(u \in Z\), \(\type_T^S(u) = S_1\) (by definition). Let \(Z'\) denote the subset of vertices in \(Z\) that have at least \(k+1\) in-neighbors and at least \(k+1\) out-neighbors within \(Z\). We define the set \(Z \setminus Z'\) as \(\widetilde{Z}\). We begin with the following observation.

\begin{observation}
	\label{lem:bounded_deg_vertices}
	In a tournament $\mathcal{T}$, the  number of vertices with in-degree (resp., out-degree)  at most $d$ is upper bounded by $2d+1$.
\end{observation}

\begin{proof}
	We will establish an upper bound specifically for the vertices with an in-degree of at most \(d\). The reasoning for vertices with a bounded out-degree follows similarly. Let \(Y\) be the set of vertices with an in-degree of at most \(d\). If we denote \(|Y| = c\), we find that $\sum_{y \in Y} \indeg(y) \leq cd$. Additionally, \(\mathcal{T}[Y]\) forms a tournament on \(c\) vertices, which means there are \(\binom{c}{2}\) arcs present in it. Each arc in \(\mathcal{T}[Y]\) contributes to one in-degree, leading to the inequality 
	$\sum_{y \in Y} \indeg(y) \geq \binom{c}{2}$. Combining these two results, we can deduce that $cd \geq \binom{c}{2}$, which implies that $
	c \leq 2d + 1$.
\end{proof}

Recall that every vertex \( v \in \widetilde{Z} \) has either an in-degree or an out-degree of at most \( k \) within \( Z \). \cref{lem:bounded_deg_vertices} directly imply the followings.

\begin{observation} \label{lem:bounded_deg_vertices1}
	The number of vertices in $\widetilde{Z} (= Z \setminus Z')$ is upper-bounded by $4k+2$. 
\end{observation}

Now we divide $\widetilde{Z}$ into two sets $\widetilde{Z}_-$ and $\widetilde{Z}_+$ where $\widetilde{Z}_-$ consists of the vertices with in-degree at most $k$ in $Z$, and $\widetilde{Z}_+$ consists of the vertices with an out-degree at most $k$ in $Z$. Further let $R = V(T) \setminus (S \cup Z)$ be the set of remaining non-$S$ vertices. We can now categorize the vertices in $R$ as follows:

\begin{mdframed}[backgroundcolor=gray!10,topline=false,bottomline=false,leftline=false,rightline=false] 
	\begin{description}
		\item[{$R_<$}{\label{cat1}}:] Set of vertices in $R$ with in-degree at most $k$ and out-degree at least $k+1$ in $Z'$
		
		\item[{$R_>$}{\label{cat2}}:] Set of vertices in $R$ with in-degree at least $k+1$ and out-degree at most $k$ in $Z'$
		
		\item[{$R_+$}{\label{cat3}}:] Set of vertices in $R$ with in-degree and out-degree at least $k+1$ in $Z'$
		
		\item[{$R_-$}{\label{cat4}}:] Set of vertices in $R$ with in-degree and out-degree at most $k$ in $Z'$
	\end{description}
\end{mdframed}

Furthermore, we want to show that $R = R_<  \uplus R_>$, in other words $R_+ = R_-= \emptyset$. We will prove this in \Cref{claim:Rempty}.  Before that, we define the notion of \surearc, which will be extensively used in our kernelization algorithm, followed by an observation. 

\begin{definition}[\surearc]\label{def:sure}
	{\em For an instance $(T,S,k)$ of \sfastr, an arc $e \in A(T)$ is called a \surearc  if every  solution of size at most $k$ avoids   $e$.}
\end{definition}

The intuition behind \surearcs is that there is no  solution of size at most $k$ in $T$ that use any  \surearc (see \Cref{fig:surearc} for an llustration).  The following lemma shows a sufficient condition for being a  sure arc.  

\begin{figure}[ht!]
	\begin{center}
		\includegraphics[scale=0.55]{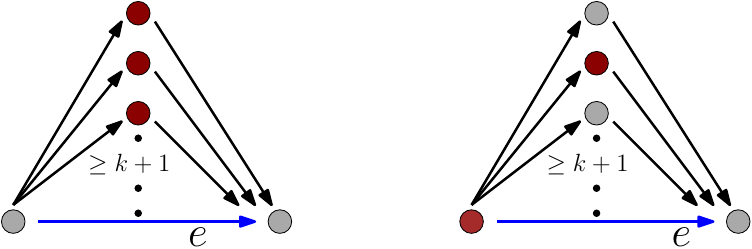}
	\end{center}
	\caption{Illustration of a \surearc $e$. Vertices with color red are from   $S$.}
	\label{fig:surearc}
\end{figure}

\begin{lemma}[Sure Arc Lemma]\label{lem:sure}
	Let $(T,S,k)$ be an instance  of \sfastr. Consider an arc  $e=(u,v)$ in $T$.  If there are more than  $k$ arc disjoint $S$-paths from $u$ to $v$, then $e$ is a {\em \surearc} for $(T,S,k)$. 
\end{lemma}

\begin{proof}
	Suppose that there is an arc $e=(u,v) \in A(T)$ such that there are  more than $k$ arc disjoint $S$-paths from $u$ to $v$. Consider any solution $F$ of size at most $k$. Observe that if $e \in F$, then we $F$ must hit each of the $S$-paths from $u$ to $v$ as each  such path creates a $S$-cycle along with the arc $(v,u)$. But then it exceeds our budget $k$.
\end{proof}

Next, we show that the majority of the arcs entering \( Z' \) or exiting \( Z' \) from \( R \) are \surearc. Recall that $\arc(X,Y)$ denotes the set of arcs with one endpoint in $X$ and the other endpoint in $Y$. We refer to the arcs in $\arc(X,Y)$ as being from $X$ to $Y$ if each arc originates from $X$ and terminates at $Y$. 
\begin{observation} \label{obs:sure_arc2}
	Let $(T,S,k)$ be an instance of \sfastr where  Reduction Rule \ref{redrule:many_triangle} is not  applicable. Let $Z$ be an equivalence class with $|Z| \geq 6k+6$ and $Z$ partitions $S$ into $S_1 \uplus S_2$ such that for every vertex $v \in Z$, $\type_T^S(v) =  S_1$, or in other words, $Z = \eqv_T^S[S_1]$. As defined earlier, let \(Z'\) denote the subset of vertices in \(Z\) that have at least \(k+1\) in-neighbors and at least \(k+1\) out-neighbors within \(Z\). Then we have that,  
	\begin{enumerate}[(i)]
		\item Every  arc in $\arc(S_1, S_2)$ is from $S_1$ to $S_2$, moreover   is a   \surearc in $(T,S,k)$.
		\item Every arc of $\arc(S_1, Z')$ is a \surearc in $(T,S,k)$.
		
		\item  Every arc of $\arc(Z', S_2)$ is a \surearc in $(T,S,k)$.
	\end{enumerate}
	
\end{observation}

\begin{proof} (i) Assume towards contradiction that there is an arc $(v,u) \in A(T)$ such that $v \in S_2$ and $u \in S_1$. We know that every arc in $\arc(S_1, Z)$ is from $S_1$ to $Z$ in $T$ and every arc in $\arc(S_2, Z)$ is from $Z$ to $S_2$ in $T$.  Since $|Z| \geq 6k+6$, we obtain at least $(k+1)$ $S$-triangles where the arc $(v,u)$ is common in each of them. This leads to a contradiction that $(T,S,k)$ is an instance where \cref{redrule:many_triangle} can not be applied. Moreover, since $|Z| \geq 6k+6$, it follows that there exist $(k+1)$ arc-disjoint two length $S$-paths from $u$ to $v$, where each path passes through a vertex $z \in Z$. So, the arc $(u,v)$ is a \surearc by \Cref{lem:sure}.

	(ii)  Consider an arc $(u,v)$ where $u \in S_1$ and $v \in Z'$. Given that $v$ has $(k+1)$ in-neighbors in $Z$, let that set be $W$. Since $Z = \eqv_T^S[S_1]$, it follows that
	there exist $(k+1)$ arc-disjoint two length $S$-paths from $u$ to $v$, where each path passes through a vertex $w \in W$. Hence, the arc $(u,v)$ is a \surearc in $(T,S,k)$ by \Cref{lem:sure}. 
	
	(iii)  Consider an arc $(u,v)$ where $u \in Z'$ and $v \in S_2$. Given that $u$ has at least $(k+1)$ out-neighbors in $Z$, let that set be $W$. Since $Z = \eqv_T^S[S_1]$, it follows that there exist $(k+1)$ arc-disjoint two length $S$-paths from $u$ to $v$, where each path passes through a vertex $w \in W$. Hence, the arc $(u,v)$ is a \surearc in $(T,S,k)$  by \Cref{lem:sure}. 
\end{proof}

For an illustration of the \Cref{obs:sure_arc2}, we refer to \Cref{fig:surearc5}.

\begin{figure}[ht!]
	\begin{center}
		\includegraphics[scale=0.45]{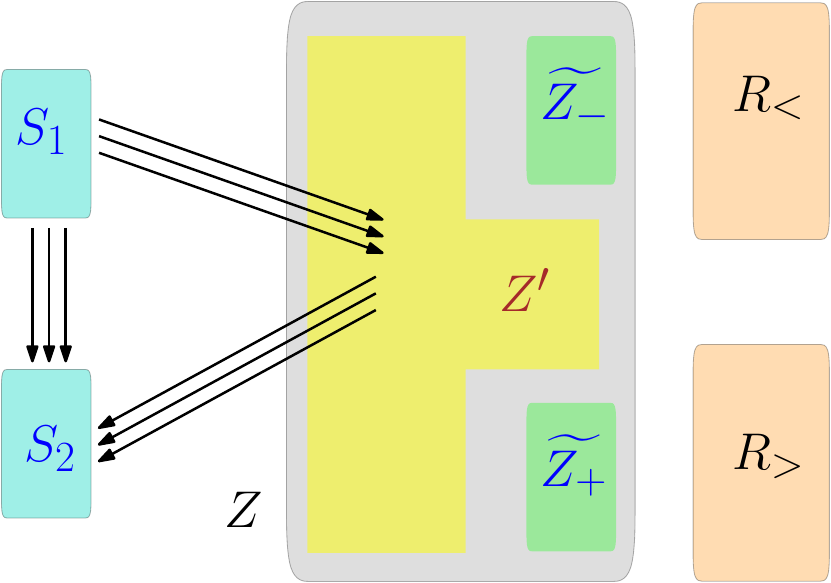}
	\end{center}
	\caption{Illustration of a \surearcs in \Cref{obs:sure_arc2}.}
	\label{fig:surearc5}
\end{figure}


	
	
	

	
	
	


\begin{lemma}\label{claim:Rempty}
	The sets $R_+$ and $R_-$ are empty.
\end{lemma}

\begin{proof}
	Consider a vertex \( v \in R_+ \). Since \( \type_T^S(v) \neq \type_T^S(u) \) for any \( u \in Z \), it follows that either i) there exists a vertex \( s_1 \in S_1 \) such that \( (v, s_1) \in A(T) \), or ii) there exists a vertex \( s_2 \in S_2 \) such that \( (s_2, v) \in A(T) \). Let's first examine the case where there is a vertex \( s_1 \in S_1 \) such that \( (v, s_1) \in A(T) \). The argument for the other case is symmetric. Given that \( v \in R_+ \), it implies that \( v \) has at least \( k+1 \) in-neighbors in \( Z' \), which we denote as the set \( X \subseteq Z' \). Additionally, we know that \( s_1 \) has an arc to every vertex \( z \in Z' \). Thus, there are \( (k+1) \) \( S \)-triangles, each sharing the arc \( (v, s_1) \). This leads to a contradiction, as it implies that \( (T,S,k) \) is an instance where \cref{redrule:many_triangle} is applicable.

	Now consider a vertex $u \in R_-$. By definition, $u$ has an in-degree and an out-degree at most $k$ in $Z'$. Now,  $|Z| \geq 6k+6$ imply that  $|Z'| \geq  2k+4$ (by \cref{lem:bounded_deg_vertices1}). Since  $T$ is a tournament, every vertex in $V(T) \setminus Z'$ has at least either $k+2$ in-neighbor or $k+2$ out-neighbor in $Z'$ (by pigeon hole principle). Hence $R_-$ cannot be non-empty.    
\end{proof}

We will further subdivide the sets \( R_< \) and \( R_> \). The process for partitioning \( R_< \) will be described, while a similar approach will be applied to \( R_> \). Let \( \widehat{R}_< \subseteq R_< \) denote the subset of vertices in \( R_< \) that does not have an in-neighbor in \( Z' \). We define \( \widetilde{R}_< = R_< \setminus \widehat{R}_< \) as the remaining vertices in \( R_< \), where each vertex in \( \widetilde{R}_< \) has at least one in-neighbor in \( Z' \). Our next reduction rule aims to demonstrate that for a \yes instance, the size of \( \widetilde{R}_< \) and $\widetilde{R}_>$ must be bounded.

\begin{reduction rule} \label{redrule:wrong_arcs}
	Let $(T,S,k)$ be an instance where none of the Reduction Rules \ref{redrule:sanity_check} - \ref{redrule:safepartition} are applicable.
	If $|\widetilde{R}_<|> k$ or $|\widetilde{R}_>| > k$, then we return a trivial \no instance.
\end{reduction rule}

It is easy to verify that \cref{redrule:wrong_arcs} can be applied in polynomial time.  In the following, we prove the correctness of it.

\begin{lemma} \label{lem:wrong_arcs}
	\cref{redrule:wrong_arcs} is sound.
\end{lemma}
\begin{proof}
	We only prove the soundness of the \cref{redrule:wrong_arcs} for the case when $|\widetilde{R}_<| > k$. The proof for the case where $|\widetilde{R}_>| > k$ is analogous.
	
	Let $(T,S,k)$ be a \yes instance for \sfastr. Let $F$ be a solution of size at most $k$. We have a corresponding equivalence class $Z$ with a size at least $6k+6$. Additionally we have $S_1, S_2, \widetilde{Z},  Z', R_<$, $\widetilde{R}_<$, $R_>$ and $\widetilde{R}_>$ corresponding to $Z$. Consider a vertex $v \in \widetilde{R}_<$. Since we know that for each $u \in Z$, $\type_T^S(v) \neq \type_T^S(u)$, we have either of two cases: i) there is a vertex $s_1 \in S_1$ such that $(v, s_1) \in A(T)$, and ii) there is a vertex $s_2 \in S_2$ such that $(s_2, v) \in A(T)$.   We analyze both the cases separately.
	
	\begin{description}
		\item[Case 1: There is a vertex $s_1 \in S_1$ such that $(v, s_1) \in A(T)$. \label{case2}] Since $v \in \widetilde{R}_<$, this implies that $v$ has at least one in-neighbors in $Z'$. Let $z' \in Z'$ such that $(z',v) \in A(T)$. We also know that $(s_1, z') \in A(T)$ (an arc in $T$). This implies that there is a triangle $\triangle$ in $T$ with arcs $(z',v), (v,s_1), (s_1,z')$. Now since $(s_1,z')$ is a \surearc in $(T,S,k)$ (by \cref{obs:sure_arc2}), this implies that any solution of size at most $k$ must avoid this arc. Hence, for the triangle $\triangle$, $F$ must contain an arc that is incident to $v$, that is at least one of the arcs from $ \{(z',v),(v,s_1)\}$.
		
		\item[Case 2: There is a vertex $s_2 \in S_2$ such that $(s_2, v) \in A(T)$.\label{case1}] Since $v \in \widetilde{R}_<$, this implies that $v$ has at least  $k+1$ out-neighbors to $Z'$, let that set be $X \subseteq Z'$. Moreover we know that every vertex $z \in Z'$ has an arc to $s_2$. Thus, there are \( (k+1) \) \( S \)-triangles, each sharing the arc \( (s_2, v) \). This leads to a contradiction, as it implies that \( (T,S,k) \) is an instance where \cref{redrule:many_triangle} is applicable.

		
	\end{description}
	
	Considering both cases, we can conclude that Case \hyperref[case1]{2} cannot occur. Given that the solution can contain at most $k$ arcs, this implies that the total number of vertices that satisfies Case \hyperref[case2]{1} is upper bounded by $k$. Hence, we can say that $|\widetilde{R}_<| \leq k$.
\end{proof}

Therefore, we can state the following lemma. 

\begin{lemma}\label{lem:Rbound}
	Let $(T,S,k)$ be an instance where none of the Reduction Rules \ref{redrule:sanity_check}-\ref{redrule:wrong_arcs} are applicable. Then $|\widetilde{R}_<| \leq k$ and $|\widetilde{R}_>| \leq k$.
\end{lemma}

For convenience, we have summarized all the sets discussed above, as well as the final one that will be defined at the end.

\begin{mdframed}[backgroundcolor=gray!10,topline=false,bottomline=false,leftline=false,rightline=false] 
	\begin{enumerate}
		
		\item[$\bullet$] $V(T) = S \uplus Z \uplus R$
		
		\item[$\bullet$] $S= S_1 \uplus S_2 $
		
		\item[$\bullet$] $Z = \{ v~:~v \notin S,~N^-(v) \cap S= S_1\} = Z' \uplus \widetilde{Z}  $
		\begin{itemize}
			\item $Z'= \{v~:~|N^+(v) \cap Z|\geq k+1,~ |N^-(v) \cap Z|\geq k+1\}$
			
			\item $\widetilde{Z}= \widetilde{Z}_- \uplus \widetilde{Z}_+ $
			
			\begin{itemize}
				\item  $\widetilde{Z}_- = \{v~:~|N^-(v) \cap Z| \leq  k \}$
				
				\item $\widetilde{Z}_+ = \{v~:~|N^+(v) \cap Z| \leq k  \}$
			\end{itemize}
		\end{itemize}
		
		\item[$\bullet$] $R= R_< \uplus R_> $
		\begin{itemize}
			\item $R_<= \{ v~;~|N^-(v) \cap Z'|\leq k,~ |N^+(v) \cap Z'|\geq k+1\} = \widehat{R}_< \uplus \widetilde{R}_<  $
			
			\begin{itemize}
				\item $\widehat{R}_< = \{v~:~|N^-(v) \cap Z'| = 0 \}$
				\item $\widetilde{R}_< = \{v~:~|N^-(v) \cap Z'|\geq 1\}$
			\end{itemize}
			
			\item $R_>= \{ v~:~|N^-(v) \cap Z'|\geq k+1,~ |N^+(v) \cap Z'|\leq k\} = \widehat{R}_> \uplus \widetilde{R}_>  $

			\begin{itemize}
				\item $\widehat{R}_> = \{v~:~  |N^+(v) \cap Z'| = 0 \}$
				\item $\widetilde{R}_> = \{v~:~|N^+(v) \cap Z'|\geq 1\}$
			\end{itemize}
		\end{itemize}  
		
		\item[$\bullet$] $Z_{\relevant} = \INN_{Z'}(R_<) \cup \OUTT_{Z'}(R_>) \cup \INN_{Z'}(\widetilde{Z}_-) \cup \OUTT_{Z'}(\widetilde{Z}_+)$
	\end{enumerate}
\end{mdframed}
For a pair of disjoint vertex sets \( A \) and \( B \), we define \(\INN_{B}(A)\) (respectively, \(\OUTT_{B}(A)\)) as the set of all in-neighbors (respectively, out-neighbors) of \( A \) in \( B \). Recall that \( Z' = Z \setminus \widetilde{Z} \). We now examine \(\INN_{Z'}(R_<)\). Note that \(\INN_{Z'}(R_<)\) corresponds to the in-neighbors of \(\widetilde{R}_<\) in \( Z' \), since \(\widehat{R}_<\) has no in-neighbors in \( Z' \). The size of \(\INN_{Z'}(R_<)\) is clearly upper-bounded by \( k^2 \) since the size of the set \(\widetilde{R}_<\) is at most \( k \), and each vertex in \(\widetilde{R}_<\) has at most \( k \) in-neighbors in \( Z' \). Similarly, we consider \(\OUTT_{Z'}(R_>)\). Observe that \(\OUTT_{Z'}(R_>)\) represents the out-neighbors of \(\widetilde{R}_>\) in \( Z' \), as \(\widehat{R}_>\) has no out-neighbors in \( Z' \). Thus, the size of \(\OUTT_{Z'}(R_>)\) is also upper-bounded by \( k^2 \) since the size of the set \(\widetilde{R}_>\) is at most \( k \), and each vertex in \(\widetilde{R}_>\) has at most \( k \) out-neighbors in \( Z' \).

Additionally, we define \(\INN_{Z'}(\widetilde{Z}_-)\) and \(\OUTT_{Z'}(\widetilde{Z}_+)\). The sizes of these sets, \(\INN_{Z'}(\widetilde{Z}_-)\) and \(\OUTT_{Z'}(\widetilde{Z}_+)\), are each upper-bounded by \( k(2k+1) \) (Observations~\ref{lem:bounded_deg_vertices} and \ref{lem:bounded_deg_vertices1}).

We introduce \( Z_{\relevant} \coloneq \INN_{Z'}(R_<) \cup \OUTT_{Z'}(R_>) \cup \INN_{Z'}(\widetilde{Z}_-) \cup \OUTT_{Z'}(\widetilde{Z}_+) \) as the set of relevant vertices associated with the equivalence class \( Z \). 
The size of \( Z_{\relevant} \) is bounded above by \( 2k(2k+1) + 2k^2 = 6k^2 + 2k \). At this point, we can simplify the instance by deleting any vertex that is not part of \( Z_{\relevant} \). However, this would lead to a kernel with $\OO(k^3)$ vertices.  For an illustration of the partition of the vertex set $V(T)$, in particular the set $Z_{\relevant}$, we refer to \Cref{fig:partition}.

\begin{figure}[t!]
	\begin{center}
		\includegraphics[scale=0.5]{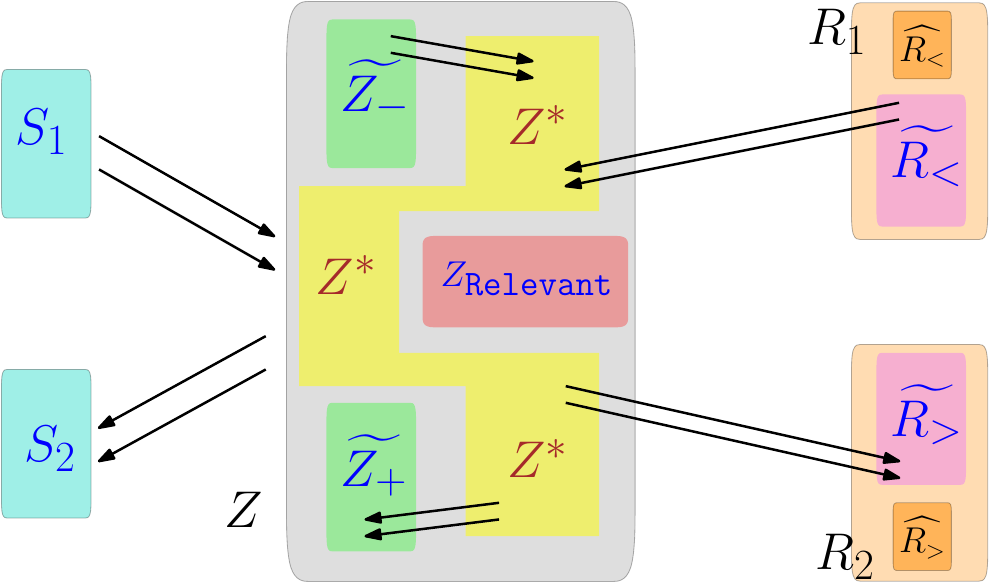}
	\end{center}
	\caption{Illustration of our partition. Double directed arrow denotes that  they have arc in that direction for each pair of vertices.}
	\label{fig:partition}
\end{figure}

In the following subsection, we give further improvement to the size of $Z_{\relevant}$.

\subsection{Linear Bound on the number of  Relevant Vertices} \label{subsec:linear_bound_eqclasses}
In this section, we provide an improved bound for the size of \( Z_{\relevant} \). Recall that 
$$ Z_{\relevant} = \INN_{Z'}(R_<) \cup \OUTT_{Z'}(R_>) \cup \INN_{Z'}(\widetilde{Z}_-) \cup \OUTT_{Z'}(\widetilde{Z}_+) $$ 
is a subset of vertices in \( Z' \). Let \(  Z_{\irrelevant} \coloneq Z' \setminus Z_{\relevant} \).

One reason why the size of \( Z_{\relevant} = \mathcal{O}(k^2) \) is that each vertex in \( \widetilde{Z}_- \) could have \( k \) in-neighbors in \( Z' \). Since \( \widetilde{Z}_- \) contains \( 2k+1 \) vertices, we have \( |\INN_{Z'}(\widetilde{Z}_-)| \leq 2k^2 + k \). We keep all these in-neighbors in \( Z_{\relevant} \) because they validate the redundancy of the vertices in \( Z_{\irrelevant} \). However, the redundancy of a vertex in \( Z_{\irrelevant} \) depends on the \emph{number of in-neighbors} and \emph{number of out-neighbors} in \( Z_{\relevant} \), not on the actual neighbors themselves. We leverage this observation and apply an ``arc swapping argument'' (see \Cref{fig:swap} for an illustration) to obtain an equivalent instance in which, for any vertex in \( Z_{\irrelevant} \), the number of in-neighbors and number of out-neighbors in \( Z_{\relevant} \) remains the same, but the size of \( Z_{\relevant} \) is reduced to \( \mathcal{O}(k) \).



\begin{figure}[t!]
	\begin{center}
		\includegraphics[scale=0.55]{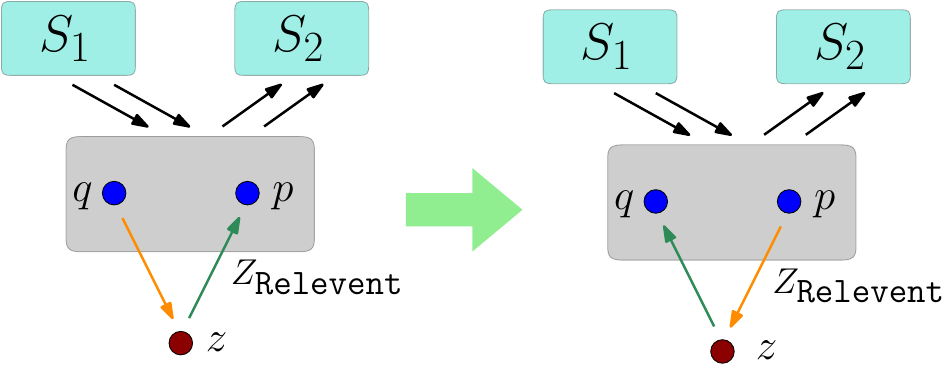}
	\end{center}
	\caption{Illustration of Arc Swapping.}
	\label{fig:swap}
\end{figure}


\begin{lemma}[Arc Swapping Lemma]\label{lem:swapping_lemma}
	Let $z \notin Z_{\relevant}$ be any non-$S$ vertex such that there exist two vertices $p,q \in Z_{\relevant}$ such that $P= \{(z,p), (q,z)\} \subseteq  A(T)$  then $(T, S, k)$ is a {\em \yes} instance if and only if $(T'= T\circledast P, S, k)$ is a {\em \yes} instance. 
\end{lemma}

\begin{proof}
	In the forward direction, let $(T, S, k)$ be a \yes instance and let $F$ be a minimal solution of size at most $k$. Now we make a case distinction as follows.
	
	\begin{description}
		\item[Case 1: $(z,p),~(q,z) \in F$:\label{case31}]  In this case, we will show that $F'= F \setminus P$ is a solution of size at most $k-2$ for $(T', S, k)$. Notice that $T\circledast F$ and $T'\circledast F'$ are the same tournaments. So, if $T'\circledast F'$ has a $S$-triangle, then that will imply that $T\circledast F$ also has a $S$-triangle which contradicts that $F$ is a solution for $(T,S,k)$. 
		
		
		\item[Case 2: $(z,p) \in F,~(q,z) \notin F$:\label{case32}] We only consider this case since the case $(z,p) \notin F,~(q,z) \in F$ is symmetric. Consider $F'= F \setminus (z,p) \cup \rev(q,z)$. We claim that $F'$ is a solution of size at most $k$ for $(T', S, k)$.  Notice that the tournaments $T\circledast F$ and $T'\circledast F'$ are the same tournaments. Hence it directly follows that any $S$-triangle in $T'\circledast F'$ will directly imply having the same $S$-triangle in $T\circledast F$ which contradicts that $F$ is a solution for $(T,S,k)$.
		
		
		\item[Case 3: $(z,p) \notin F,~(q,z) \notin F$:\label{case34}] In this case, we claim that $F$ is also a solution of size at most $k$ for $(T',S,k)$. Suppose not, this implies that $T'\circledast F$ has a $S$ triangle, $\triangle$. Since $T\circledast F$ did not have any $S$-triangle, this implies that $\triangle$ must contain at least one of the arcs $(p,z)$ and $(z,q)$ (reversed arcs corresponding to $(z,p)$ and $(q,z)$, respectively). Without loss of generality assume that $(p,z) \in \triangle$ (proof for the case when  $(z,q) \in \triangle$ is identical).
		
		\begin{description}\item[When $(p,z) \in \triangle$: \label{subcase1}] 
			Since $z,p \notin S$, let the third vertex of $\triangle$ be $s \in S$ and the arcs of $\triangle$ be $(p,z),(z,s)$ and $(s,p)$. Since $|Z| > 7k+4$ and $p \in Z_{\relevant} \subseteq Z'$, the arc $(s,p)$ is a \surearc in $(T,S,k)$ (by \cref{obs:sure_arc2}). Moreover since $p \in Z'$ and for every $s' \in S_2$, the arc $(p,s')$ is a \surearc in $(T,S,k)$ (by \cref{obs:sure_arc2}), this implies that $s \notin S_2$. So we assume that $s \in S_1$. Now consider $T\circledast F$. Since $q \in Z'$, the arc $(s,q)$ is also a \surearc in $(T,S,k)$ (by \cref{obs:sure_arc2}) and hence $F$ avoids this arc in the instance $(T,S,k)$. $T\circledast F$ has the arc $(q,z)$ since $(q,z) \notin F$. Consider the arc $(z,s)$. If $(z,s) \notin F$, then we found a triangle $(s,q),(q,z)$ and $(z,s)$ in $T\circledast F$ which contradicts that $F$ is a solution for $(T,S,k)$. Let's assume that $(z,s) \in F$. Since $F$ is a minimal solution, there must exist a  cycle $C$ such that $\arc(C) \cap F = \{(z,s)\}$. Now, consider a closed walk \( W \) that traces the triangle formed by the arcs \( (s,q), (q,z), \) and \( (z,s) \). The arcs \( (s,q) \) and \( (q,z) \) remain the same in $W$, but instead of following the arc \( (z,s) \), we traverse the path \( C - (z,s) \) from the corresponding cycle of \( (z,s) \). By definition, \( W \) is a closed walk, and notably, it contains the vertex \( s \in S_1 \) while avoiding all arcs from \( F \). From this, we can easily extract a cycle that includes \( s \), and by \cref{lem:Scycle}, we find a triangle containing \( s \) that avoids all arcs of \( F \), leading to a contradiction.    
		\end{description}
	\end{description}
	In the reverse direction, note that we apply the same operation on the instance $(T', S, k)$ to transform it back to $(T, S, k)$. Thus, the proof for the reverse direction mirrors the proof for the forward direction.
\end{proof}

The next lemma utilizes the Arc Swapping Lemma (\Cref{lem:swapping_lemma}) to reduce the size of \(Z_{\relevant}\) to \(2k+2\).

\begin{lemma}\label{lem-upper-bound-relevant}
	Given an instance $(T,S,k)$ of \sfastr, in polynomial time we can obtain an instance $(T',S,k)$ such that 
	$(T,S,k)$ is a \yes instance if and only if $(T',S,k)$ is a \yes instance. Moreover, in $T'$, the size of $Z_{\relevant}$ is upper bounded by $2k+2$.
\end{lemma}
\begin{proof}
	We assume that the size of \(Z_{\relevant}\) is at least \(2k+2\); otherwise, the proof is already complete.  
	We partition \( Z_{\relevant} \) into two parts, \( Z_{\relevant}^1 \uplus Z_{\relevant}^2 \), where \( Z_{\relevant}^1 = \INN_{Z'}(\widetilde{Z}_-) \cup \INN_{Z'}(R_<) \) and \( Z_{\relevant}^2 = \OUTT_{Z'}(\widetilde{Z}_+) \cup \OUTT_{Z'}(R_>) \). We will first provide an upper bound for the size of the set \( Z_{\relevant}^1 \), and the argument for bounding the size of \( Z_{\relevant}^2 \) is symmetric. The proof utilizes Lemma~\ref{lem:swapping_lemma}, applied to a specific subset of \( Z_{\relevant}^1 \).
	
	We order the vertices in \( Z_{\relevant}^1 \) arbitrarily. Let \( Z_{\relevant}^1[k+1] \) denote the first \( k+1 \) vertices in that order. Now, consider the set \( \widetilde{Z}_- \cup R_< \), which contains vertices with in-degree at most \( k \) in \( Z' \). For each vertex \( v \in \widetilde{Z}_- \cup R_< \), proceed as follows: choose an in-neighbor \( u \) of \( v \) in \( Z_{\relevant}^1 \setminus Z_{\relevant}^1 [k+1] \). Since the number of in-neighbors of \( v \) in \( Z' \) is at most \( k \), there must be a vertex \( w \in Z_{\relevant}^1[k+1] \) such that \( w \) is an out-neighbor of \( v \). Given that \( u \) is an in-neighbor and \( w \) is an out-neighbor of \( v \), we can apply \cref{lem:swapping_lemma} to obtain an equivalent instance \( (T' = T - \{(u,v),(v,w)\} + \{(v,u),(w,v)\}, S, k) \). 
	
	We continue applying \cref{lem:swapping_lemma} for each vertex \( v \) until \( v \) has no remaining in-neighbors in \( Z_{\relevant}^1 \setminus Z_{\relevant}^1 [k+1] \). Then, we move to the next vertex in \( \widetilde{Z}_- \cup R_< \). The proof of the correctness of this step follows by induction on the number of swapped arcs and from the correctness of the lemma \cref{lem:swapping_lemma}. Once this process is applied exhaustively to every vertex in \( \widetilde{Z}_- \cup R_< \), the arcs in \( A(\widetilde{Z}_- \cup R_<, Z_{\relevant}^1 \setminus Z_{\relevant}^1 [k+1]) \) will all originate from \( \widetilde{Z}_- \cup R_< \) and terminate in \( Z_{\relevant}^1 \setminus Z_{\relevant}^1 [k] \). Furthermore, all the in-neighbors of \( \widetilde{Z}_- \cup R_< \) within \( Z' \) will belong to \( Z_{\relevant}^1 [k+1] \), i.e., \( \INN_{Z'}(\widetilde{Z}_-) \cup \INN_{Z'}(R_<) \subseteq Z_{\relevant}^1 [k+1] \). Since \( Z_{\relevant}^1 [k+1] \) has size \( k+1 \), the size of \( Z_{\relevant}^1 = \INN_{Z'}(\widetilde{Z}_-) \cup \INN_{Z'}(R_<) \) is also bounded by \( k+1 \).
	
	Similarly, the size of \( Z_{\relevant}^2 \) is also bounded above by \( k+1 \). Therefore, the total size of \( Z_{\relevant} \) is at most \( 2k+2 \). This completes the proof.
\end{proof}

In the next subsection, we show that any vertex in \(Z_\irrelevant\) can be safely removed until the size of \(Z'\) is reduced to \(2k+4\).

\subsection{Finding Irrelevant Vertex from Large Equivalence Class}
Recall that $Z_{\irrelevant} ~ \coloneq Z'~ \setminus Z_{\relevant}$. We call every vertex of the set $Z_{\irrelevant}$, an irrelevant vertex. All the vertices of $Z_{\irrelevant}$  have the following properties (by \cref{obs:sure_arc2}) which is captured in the following observation.

\begin{observation}\label{obs:surearc-irv}
	\begin{mdframed}
		[backgroundcolor=gray!10,topline=false,bottomline=false,leftline=false,rightline=false] 
		$~~$
		\begin{enumerate}
			\item  Every  arc in $\arc(\widetilde{Z}_- \cup S_1 \cup R_<, Z_{\irrelevant})$ is a \surearc in $(T,S,k)$ and those arcs are from $(\widetilde{Z}_- \cup S_1 \cup R_<)$ to $Z_{\irrelevant}$.
			
			\item  Every arc in $\arc(Z_{\irrelevant}, \widetilde{Z}_+ \cup S_2 \cup R_>)$ is a \surearc in $(T,S,k)$ and those arcs are from $Z_{\irrelevant}$ to $(\widetilde{Z}_+ \cup S_2 \cup R_>)$.
			
		\end{enumerate}
	\end{mdframed} 
\end{observation}

In the subsequent reduction rule, our objective is to eliminate an irrelevant vertex to establish a linear (in \(k\)) bound on each equivalence class. 
We formalize this as follows.


\begin{reduction rule}[Irrelevant Vertex Rule 2]\label{redrule:bigtype}
	Let \((T,S,k)\) be an instance of \sfastr where none of the Reduction Rules \ref{redrule:sanity_check}-\ref{redrule:wrong_arcs} are applicable. Furthermore, this instance is obtained after applying \Cref{lem-upper-bound-relevant}. Let \(Z_\irrelevant\) be the set defined above. If \(|Z| \geq 6k+7\), then delete an arbitrary vertex \(v\) from \(Z_\irrelevant\) in \(T\). Return the instance \((T - v, S, k)\).
\end{reduction rule}

It is easy to verify that \cref{redrule:bigtype} can be applied in polynomial time.  In the following, we prove the correctness of it.

\begin{lemma}
	\cref{redrule:bigtype} is sound.
\end{lemma}

\begin{proof}
	In the forward direction, let \((T,S,k)\) be a \yes instance of \sfastr and \(F \subseteq A(T)\) be a solution of size at most $k$. Since \(T-v\) is an induced subgraph of \(T\), \(F' = F \setminus \arc(v)\) is a solution of size at most $k$ for \((T - v, S, k)\).
	
	In the reverse direction, let \((T - v, S, k)\) be a \yes instance for \sfastr and \(F' \subseteq A(T-v)\) be a solution of size at most $k$. Let \(Z_v = Z \setminus \{v\}\) and \(Z'_v = Z' \setminus \{v\}\). Observe that since \(|Z_v| \geq 6k + 7\), it follows that \(|Z'_v| \geq 2k + 5\).
	
	We first show that \(F'\) must avoid every arc from the set \(\arc(S_1, Z'_v)~ \cup~ \arc(Z'_v, S_2)~ \cup ~\arc(S_1, S_2)\) in $(T -v,S,k)$. 
	Note that since $|Z| \geq 6k+7$, this implies that $|Z_v| \geq 6k+6$. Hence, 
	by \cref{obs:sure_arc2}, we can say that every arc in \(\arc(S_1, Z'_v) \cup \arc(Z'_v, S_2) \cup \arc(S_1, S_2)\) is a \surearc in \((T - v, S, k)\). We now prove the following claim.

	\begin{claim}
		$F'$ is a solution  of \sfastr on $(T,S,k)$.
	\end{claim}
	
	\begin{claimproof}
		Assume towards contradiction that $F'$ is not a solution in $T$. Then  $T\circledast F'$ must have  a $S$-triangle $\triangle$. Now since, $F'$ was a solution for $(T - v,S,k)$, the triangle $\triangle$  must contain the  vertex $v$ and an $S$-vertex, say $s$. Let $z$ be the other vertex of $\triangle$ other than $\{v,s\}$. Now we can have one of the following cases depending on the arcs of the triangle. 
		\begin{description}
			\item[Case A: $(s,v) \in \triangle \ \&\  s \in S_1:$ \label{case11}] In this case, the arcs of the triangle, $\triangle$ are $(s,v),(v,z)$ and $(z,s)$. Firstly $z \notin Z'_v$ because every arc in $\arc(S_1, Z')$ is a \surearc in $(T, S, k)$ (by \cref{obs:sure_arc2}) and $(z,s)$ is an arc of $\triangle$. Moreover, $z \notin S_2$ because every arc in $\arc(S_1, S_2)$ is a \surearc in $(T, S, k)$ (by \cref{obs:sure_arc2}) and $(z,s)$ is an arc of $\triangle$. Furthermore, $z \notin S_1$ because every arc in $\arc(S_1, Z')$ is a \surearc in $(T, S, k)$(by \cref{obs:sure_arc2}) and $(v,z)$ is an arc of $\triangle$.
			Further $z \notin (\widetilde{Z}_- \cup R_<)$. Assume towards contradiction that $z \in (\widetilde{Z}_- \cup R_<)$. Since $v \in Z_\irrelevant$, there is an arc from every vertex of $(\widetilde{Z}_- \cup R_<)$ to every vertex of $Z_\irrelevant$ in $(T,S,k)$ (by \cref{obs:surearc-irv}). But we have that $(v,z)$ is an arc of $\triangle$ with $v \in Z_{\irrelevant}$ and $z \in (\widetilde{Z}_- \cup R_<)$, which is a contradiction. Lastly, $z \notin (\widetilde{Z}_+ \cup R_>)$. Assume towards contradiction that $z \in (\widetilde{Z}_+ \cup R_>)$. Note that every vertex of $(\widetilde{Z}_+ \cup R_>)$ has at least $(k+1)$ in-neighbors in $Z_\irrelevant$, let that set of vertices be $X \subseteq Z_\irrelevant$. Now it follows that there exist $(k+1)$ arc-disjoint two length $S$-paths from $s$ to $z$, where each path passes through a vertex $x \in X$. Hence, the arc $(s,z)$ must be a \surearc in $(T,S,k)$ (by \cref{obs:sure_arc2}). But we have that $(z,s)$ is an arc of  $\triangle$ with $z \in (\widetilde{Z}_+ \cup R_>)$ and $s \in S_1$ which is a contradiction.



			
			\item[Case B: $(v,s) \in \triangle \ \&\  s \in S_2:$ \label{case12}]  The argument for this case is analogous to Case \hyperref[case11]{A}.
			
			\item[Case C: $(s,v) \in \triangle \ \&\  s \in S_2:$ \label{case12}]  This case is not possible since the arc $(v,s)$ is a \surearc in $(T,S,k)$.
			\item[Case D: $(v,s) \in \triangle \ \&\  s \in S_1:$ \label{case12}]  This case is not possible since the arc $(s,v)$ is a \surearc in $(T,S,k)$.
		\end{description}
	\end{claimproof}
	This concludes the proof. 
\end{proof}

So we have the  following Lemma.

\begin{lemma}\label{lem:ZZbound}
	Let $(T,S,k)$ be an instance where none of the Reduction Rules \ref{redrule:sanity_check}-\ref{redrule:bigtype} are applicable. Then size of each equivalence class $\eqv_T^S[X]$ is at most $(6k+6)$.
\end{lemma}

\section{Wrapping Up: Final Bound on the Kernel}

Let $(T,S,k)$ be an instance where none of the previous reduction rules are applicable. Now we have our final reduction rule.

\begin{reduction rule}\label{redrule:vertex_bound}
Let $(T,S,k)$ be an instance where the Reduction Rules \ref{redrule:sanity_check}-\ref{redrule:bigtype} are not applicable. If $|V(T)| > (30k^2 + 40k+ 7)$, then return a trivial \no instance.
\end{reduction rule}

It is easy to verify that \cref{redrule:vertex_bound} can be applied in polynomial time.  In the following, we prove the correctness of it.

\begin{lemma}
\cref{redrule:vertex_bound} is sound.
\end{lemma}

\begin{proof}
Suppose $(T,S,k)$ is a \yes instance. Since none of the previous reduction rules are applicable, $T$ has the following properties: 
\begin{enumerate}
	\item the number of $S$ vertices is upper bound by $4k$ (\Cref{lem:ter2}),
	\item the number of equivalence classes is upper bound by $(5k+1)$ (\Cref{lem:ter3}),
	\item The size of each equivalence class is also upper bound by $(6k+6)$. (\Cref{lem:ZZbound})
\end{enumerate}
So, overall number of vertices in $T$ is upper bounded by $30 k^2 + 40k+6$. 
\end{proof}

Let $(T,S,k)$ be an instance where none of the  Reduction Rules \ref{redrule:sanity_check} - 
\ref{redrule:vertex_bound} are applicable, then number of vertices of $T$ is upper bounded by $\mathcal{O}(k^2)$. Hence, we finally obtain the following theorem. 

\faskernel*


\section{FPT algorithm for \sfast} \label{sec:sfasfpt}

In this section, we prove the following result.

\sfastheorem*

Our algorithm follows the approach of Alon, Lokshtanov, and Saurabh \cite{DBLP:conf/icalp/AlonLS09} who achieve a sub-exponential time {\sf FPT} algorithm for {\sc FAST}. However, due to the inherent generality of our problem, we need to deviate  from it while implementing the outline.  Let we are given a tournament $T$, a vertex set $S \subseteq V(T)$, a positive integer $k$. We use $(T,S,k)$ to denote an instance of our problem. Our algorithm (refer as \texttt{ALGO$\_$SFAST}) for \sfast consists of three steps. 

\begin{figure}[ht!]
\noindent\fbox{
	\begin{minipage}{0.96\textwidth}
		
		\begin{description}
			\item[Step 1:] Do kernelization and obtain a tournament $T'$ of  size $\OO(k^2)$. 
			
			\item[Step 2:] Color the vertices of $ T' $ uniformly at random   with  colors from $\{1,2,\ldots,\sqrt{8k}\}$.

			\item[Step 3:] Find  a $S$-fas $ O $ of size at most $ k $ such that no arc of $ O$ is monochromatic.

		\end{description}
		
	\end{minipage}
	
}
\caption{Algorithm for \sfast.}
\label{fig:algo}
\end{figure}

Before giving the formal description of the proof, below we provide a concise overview of our algorithm for \sfast. 

\subparagraph{Overview of proof for \cref{theo:fast}.}
We begin with a polynomial kernel. Let $\Pi$ be an $S$-topological ordering corresponding to $G \oplus F$ where $F$ is a minimal solution we desire to find. Additionally, let $\Pi$ represent an ordering in which strongly connected components are contracted, ensuring a unique topological ordering. The coloring of the graph is obtained using the coloring lemma of Alon et al.~\cite{DBLP:conf/icalp/AlonLS09}. This guarantees that solution arcs appear only between different color classes. Consequently, when considering the subgraph induced by each color class, every strongly connected component remains within a strongly connected component of $G \oplus F$. This allows us to contract each strongly connected component and treat it as a single vertex. From this point onward, we assume that all strongly connected components within each color class are contracted.  Under this assumption, we obtain an ordering of $G_i$ such that when we restrict $\Pi$ to this subgraph, it preserves the same ordering. By combining these known orderings, we construct $\Pi$. To achieve this, we follow the reasoning below.  Let $w$ be the middle vertex of $\Pi$ (in the contracted representation). Our goal is to determine which contracted vertices lie to its left and which lie to its right in $\Pi$. If this can be achieved using a branching process with subexponential dependence on $k$, it leads to a subexponential-time algorithm.  To accomplish this, we first guess the last $S$ vertices from each color class that appear before $w$. However, this does not immediately create two independent subproblems because we do not yet know which arcs originate from the right side while having their tails in the left side (i.e., backward arcs that contribute to the solution). To address this, we also guess the last vertex $s$ before $w$ and determine which vertices from each color class appear after both $s$ and $w$. Specifically, we select one vertex from each color class, identifying the last vertex in each class before $s$.  

Since the number of possible choices for each of these selections is bounded by $k^{\mathcal{O}(1)}$, the resulting recurrence leads to a subexponential-time algorithm. We further optimize this approach through a carefully designed dynamic programming algorithm with the help of a set of potential integer vectors.

\subsection{Color the Vertices in a Good Way}
We want to color the vertices in such a way that for an \yes-instance $ (T, S,k) $ there is a $ S $-feedback arc set $ O $ of size at most $ k $ such that  no arc of $ O $ is monochromatic. We call such a coloring as {\em good coloring}. The following known result tells that if we randomly color the vertices of a  $k$ edge graph $G$ with $ \sqrt{8k} $ colors, then the probability that $ G $ has been properly colored is at least $ 2^{-c\sqrt{k}} $, where $c$ is a  fixed constant. 
.

\begin{proposition}[Lemma 2 \cite{DBLP:conf/icalp/AlonLS09}]\label{propo:chromatic}
If a graph on $m$ edges is colored randomly with  $ \sqrt{8m} $ colors then the probability that $ G $ is properly colored is at least $ (2e)^{-\sqrt{m/8}} $.
\end{proposition}

This proposition directly leads to the following observation.

\begin{observation}\label{obs:goodcoloring}
Let $(T,S,k)$ be an instance of \sfast and $O$ be a solution. If we color the vertices in $T$ uniformly at random with $ \sqrt{8k} $ colors	 then 
$$ \mathbb{P}r[\text{no arc of $O$ is monochromatic}] \geq  (2e)^{-\sqrt{k/8}}. $$

\end{observation} 

We say that an arc set $F \subseteq A(T)$ is {\em colorful} if no arc in $F$ is monochromatic. From now onwards, we assume that our instance $(T,S,k)$  of \sfast is given with a good coloring with $q=\sqrt{8k}$ colors. That is, we are given  a $ q $-colored tournament $T$, a vertex set $S \subseteq V(T)$, an integer $k$ as input and our goal is to check whether there exists a  colorful $S$-feedback arc set of size at most $k $, or conclude that no such $S$-feedback arc set exists. We call this problem as \csfast. We use $(T,S,k,q)$ to denote an instance of \csfast. 

\defparprob{{\sc Colorful Subset Feedback Arc Set  }(\csfast)}{A tournament $T=(V, A)$, a vertex subset $S \subseteq V$, a color function $c:V \to  [q]$,   \hspace*{12mm} and  an    integer $ k$.}{$k, q$}{$~~$Find  an arc subset $ F \subseteq A$ with $|F| \leq k$ such that   $ T\circledast F $ has no $S$-cycle and no \hspace*{10mm} $~$ arc of $F$ is  monochromatic.}

\subsection{Solving a Colored Instance}

Consider an ordering $ \sigma $ of $V(T)$. The  $S$-fas corresponding to $ \sigma $ is the set of all $ S $-backward arcs. We say  $ \sigma $ is {\em colorful ordering} if the    $S$-fas corresponding to $ \sigma $ is colorful.   A colorful ordering $\sigma$ is called  {\em valid colorful ordering} if the $S$-fas with at most $k$ many  $ S $-backward arcs. 

We now give an algorithm for  \csfast on the instance $(T,S,k,q)$. Let  $V(T) = V_1 \uplus \ldots \uplus V_q$, where for each $i \in [q]$ the set $V_i$ denotes the set of vertices with color $i$. We define $ O $  to be a  minimal solution (hypothetical) of \csfast  that we are looking for.  Clearly, if $(T,S,k)$ is an  \yes instance then $ |O| \leq k $. Let $ \sigma $ be the $S$-topological ordering of the vertices in $ T \circledast O $, respectively. That means  there is no $S$-backward arc  the ordering $\sigma$.  We use the term \textsf{\scc}  to denote strongly connected component. 

\begin{observation}\label{obs:color1}
If there exists a colorful $ S $-fas of $T$ then  for each $i \in [q]$, the subgraph $T [V_i ]$ is $S$-acyclic.
\end{observation}


Let $S_i=V_i \cap S$. Due to \Cref{obs:color1}, the vertices $V_i$ has a unique $\spclorderi$ -- we refer this ordering as $\sigma_i$. 

\begin{observation}\label{obs:color2}
For every color class $V_i$ of $T$, let $s_i^1, s_i^2, \ldots, s_i^{n_i}$ be the order in
which the vertices of $S_i$ appear according to $\sigma_i$. If $F$ is a colorful solution and $\sigma$ is the unique $\spclorder$ of  $T \circledast F $ then
\begin{enumerate}[(i)]
	\item For each $j \in [n_i -1]$, $s_i^j <_{\sigma} s_i^{j+1}$
	\item For each non $S$-vertex $v$ in $V_i$, $s_i^j <_{\sigma_i} v <_{\sigma_i}  s_i^{j+1} \implies s_i^j <_{\sigma} v <_{\sigma}  s_i^{j+1}$
\end{enumerate}
\end{observation}

\subparagraph{Modification of the instance.}

We modify the instance $ (T,S,k,q) $ of \csfast to $ (T',S,k,q) $ of the same problem as follows.
\begin{enumerate}
\item For each $ i \in [q] $, we partition the graph $ T[V_i] $ into a set of maximal \scc's.

\item We construct a new tournament  $ T' $ as follows. Let $\mathcal{C}_i$ denote the set of all maximal \scc's in $T[V_i]$ and  $\mathcal{C}\coloneqq \bigcup_{i \in [q]} \mathcal{C}_i$.
\begin{itemize}
	\item  Each \scc $ C $ in $ \mathcal{C} $ corresponds to a vertex $ v_C $ in $ T'$. 
	\item For each pair of vertices $u \in V(C_1)$ and $v \in V(C_2)$ for some $C_1, C_2 \in \mathcal{C}$, $C_1 \neq C_2$  with $(u,v) \in A(T)$, we add an arc $ (v_{C_1}, v_{C_2}) $ to $A(T')$. 
	
	\item Here we allow multiple arcs but no self loop.
	
\end{itemize}

\item For each $ i \in [q] $, as all the \scc's of $\mathcal{C}_i$ forms a DAG so  there exists a unique topological ordering of the vertices corresponding to $\mathcal{C}_i$. We denote this ordering as $ \sigma_i $.

\item For each $C \in \mathcal{C}$, we color the vertex $v_C$ with the same color as the color class, the \scc $C$ belongs to. In other words, color of $v_C$ is $i$ if $C \in \mathcal{C}_i$ in the instance \csfast.


\end{enumerate}

Observe that for each vertex $s \in S$, the graph $T[s]$ forms a strongly component in $\mathcal{C}$. Notice that $|V(T')| \leq |V(T)|$ and  $|A(T')| \leq |A(T)|$. Also, every non-monochromatic arc corresponds to some arc in $T'$. Now  in $T'$, our goal  is to find $ S $-fas $Z' \subseteq A(T')$ of size at most $k $ such that no arc of $Z' $ is monochromatic. Below in \cref{lem:equi}, we  demonstrates that the instances $(T,S,k,q)$ and  $(T',S',k',q)$ are equivalent to the problem \csfast.


\begin{lemma}\label{lem:equi}
$(T,S,k,q)$ is a {\sc Yes} instance for \csfast if and only if $(T',S,k,q)$ is a {\sc Yes} instance for \csfast.
\end{lemma}

\begin{proof}
In the forward direction, assume that  $F$ is  a solution to $(T,S,k,q)$. We can assume that $F$ is minimal. By the property of our problem, no arc of $F$ is monochromatic, so each arc in $F$ corresponds to some arc in $T'$.  Let $F'$ be the set of arcs in $T'$ corresponding to $F$ in $T$. We claim that $F'$ is also a solution to $(T',S',k',q)$. If not, then there exists a $S$-cycle, say $C'$ in $T'-F'$. Now, considering $C'$, we can retrieve a $S$-cycle $C$ in $T-F$ in the following way. Pick an arc $e=(v_{C_1}, v_{C_2} )$ in the cycle $C'$. If the arc $e$ is monochromatic in $T'$, then we get an arc $(u,v)$ for all pair of vertices  $u \in V(C_1)$ and $v \in V(C_2)$ in $T-F$ (as $F$ uses no monochromatic arc). If the arc $e$ is not monochromatic  in $T'$, then we get a directed path from $u$ to $v$ for all pair of vertices  $u \in V(C_1)$ and $v \in V(C_2)$ in $T-F$. Combining all these path and arcs we get a closed walk passing through $S$-vertices in $T-F$. For that closed walk, we can retrieve a $S$-cycle in $T-F$. This is in contradiction to the fact that $F$ is a solution to $(T,S, k, q)$.

The backward direction is analogous to the forward direction.  
\end{proof}

\medskip 

From now onwards,  we are working on the problem \csfast to the instance $(T',S, k, q)$.   For the ease of notation, once the context is clear we use $(T, S,  k, q)$ instead of $(T', S', k', q)$.  For the sake of simplicity we use the same notation which is as follows. For each $i \in [q]$, we use $V_i$ to denote the set of vertices corresponding to color $i$. 

For an integer $t \geq 1$, we define $ V^i_t= \{v_1^i, \ldots, v_t^i\} $ and $ V_0=  V^i_0 = \emptyset$.   Given an integer  vector $\hat{p}= [a_1, a_2, \ldots, a_q] $ of length $ q $ in which the $ i^{\text{th}}$ entry is between $0$ and $n_i$ where $ |V_i|=n_i $,  let $V(\hat{p})= V^1_{a_1} \cup V^2_{a_2} \cup \ldots \cup V^q_{a_q}$. Also let $ T (\hat{p})= T [V(\hat{p})]$ and $A(\hat{p})$ be the arc set of $ T (\hat{p})$. We define $S_p \coloneqq  S \cap V(\hat{p}) $. Let $\mathscr{Z}$ denotes the set of all possible vectors $\{ [a_1, a_2, \ldots, a_q]; a_i \in [n_i], i \in [q] \}$. Clearly, $|\mathscr{Z}| \leq k^{\OO(q)}$.


Our algorithm takes a $q$-colored tournament $T = (V_1 \cup V_2 \cup \ldots \cup V_q, A)$, a vertex subset $S \subseteq V(T)$  of $T$ as input, and produces a   $ S $-fas $F$  such that $|F| \leq k$    no arc of $F$  is  monochromatic. We  define the arc set $A_c$ to be the set of arcs whose endpoints have different colors.   Hence
$T[V_i]$ must be an $S$-acyclic tournament for every $i$.  Let $n_i = |V_i |$ for every
$i$ and let $\hat{n}$ be the vector $[n_1, n_2, \ldots, n_q]$. For every color
class $V_i$ of $T$ , let $v_1^i,  v_2^{i}, \ldots , v_{n_i}^i$ be the order in which the vertices of $V_i$ appear according to $\spclorder$ of  $T[V_i]$. We exploit this fact to give a dynamic programming algorithm for the problem. 







\subparagraph{Dynamic programming}

\begin{lemma}\label{lem:DP}
Given a feasible $ q $-colored tournament $ T $ of size $\mathcal{O}(k^2)$ and vertex subset $ S \subseteq V(T)$, we can find a minimum size
colorful $ S $-feedback arc set in  $k^{\OO(q)}$ time.

\end{lemma} 

\begin{proof}
We are given a $q$-colored  tournament $T$ size $\mathcal{O}(k^2)$ and  we have to find a a   $S$-fas $F$ of minimum size such that $F$ is colorful. Consider an integer   vector $\hat{p}= [a_1, a_2, \ldots, a_q] \in \mathscr{Z}$, where $a_i \in [|V_i|=n_i]$. Recall that $V(\hat{p})= V^1_{a_1} \cup V^2_{a_2} \cup \ldots \cup V^q_{a_q}$. $ T (\hat{p})= T [V(\hat{p})]$ and $A(\hat{p})$ be the arc set of $ T (\hat{p})$. Also  $S_p \coloneqq  S \cap V(\hat{p}) $.

Let $\sfas [T] $ be the size of the minimum colorful $ S $-feedback arc set of $ T $. Observe that if a $ q $-colored tournament $ T $ is feasible then so are all induced sub-tournaments of $ T $, and hence the function $\sfas$ is well defined on all induced sub-tournaments of $ T $. For a pair vectors $\hat{p}= [a_1, a_2, \ldots, a_q] $ and $\hat{q}= [b_1, b_2, \ldots, b_q] $ from $\mathscr{Z}$, we say $ \hat{q} $ respects $ \hat{p}$ (denoted by $ \hat{q} \perp \hat{p} $) if  (i) $ b_i \leq a_i $, for all $ i \in [q] $, and (ii) $V(\hat{p}) \smallsetminus V(\hat{q})$ contains no vertex  from $S$. For $\hat{q} \perp \hat{p}$, let  $A^{\gets}(\hat{q}, \hat{p}) \coloneqq  \bigl\{ (u,v)~\colon~u \in V(\hat{p}) \smallsetminus V(\hat{q}),~ v \in V(\hat{q}) \bigr\} $. Our base case occurs when $V(\hat{p}) \cap S = \emptyset$. In this case, we have $ {\sfas}[T(\hat{p})]=0$. We proceed to prove that the following recurrence holds for $\sfas[T(\hat{p})]$ when $V(\hat{p}) \cap S \neq  \emptyset$.

\begin{equation}\label{equ:onee}
	{\sfas}[T(\hat{p})]= \min \biggl\{ \min_{\hat{q} : \hat{q} \perp \hat{p}} \Bigl\{{\sfas}[T(\hat{q})]+ |A^{\gets}(\hat{q}, \hat{p})|\Bigr\} ,  \min_{s \in V(\hat{p})} \Bigl\{{\sfas}[T(\hat{p}-s)]+ |N^{+}(s) \cap V(\hat{p}|\Bigr\}\biggr\}
\end{equation}

\subparagraph{Correctness.}	
The correctness for base case is easy to follow. We now give a proof of correctness of  Recurrence \ref{equ:onee} which hold when $V(\hat{p}) \cap S \neq  \emptyset$.

First we prove that the left hand side is at most the right hand side. We may have two cases. Either there is a vector $\hat{q} \perp \hat{p}$ (1st term) or there is an $S$-vertex $s \in V(\hat{p})$ (2nd term) that  minimizes the right hand side. Firstly we consider that there is a vector $\hat{q} \perp \hat{p}$  that  minimizes the right hand side. Taking the ordering corresponding to the value $\sfas[T(\hat{q})]$	and appending it with $ V(\hat{p}) \smallsetminus V(\hat{q})  $ gives an ordering of of vertices of $T(\hat{p})$ with cost at most ${\sfas}[T(\hat{q})]+ |A^{\gets}(\hat{q}, \hat{p})|$. Secondly we consider that there is an $S$-vertex $s \in V(\hat{p})$  that  minimizes the right hand side. Taking the ordering corresponding to the value $\sfas[T(\hat{p}-s)]$	and appending it with $ s $ gives an ordering of of vertices of $T(\hat{p})$ with cost  ${\sfas}[T(\hat{p}-s)]+ |N^{+}(s) \cap V(\hat{p}|$.

To prove that the right hand side is at most the left hand side, take an optimal
colorful ordering $\sigma$ of $T[\hat{p}]$ and let $v$ be the last vertex of this ordering. There can be two cases: $v \in S$ or $v \notin S$. First, we consider that $v \in S$. Now when $\sigma$ restricted to $T[\hat{p}-s]$ we get an ordering corresponding to  the vertex set  $V(\hat{p}-s)$. As  the  number of $S$-backward arcs with respect to $\sigma$ with head  at $v$  is exactly   $ |N^{+}(s) \cap V(\hat{p}| $,   the value $\sfas[T(\hat{p})]$ is at least the value of ${\sfas}[T(\hat{p}-s)]+ |N^{+}(s) \cap V(\hat{p}|$. Next, we consider that $v \notin S$. As $V(\hat{p}) \cap S \neq  \emptyset$, there is an $S$-vertex $u$ such that $u <_{\sigma} v$ such that there is no $S$-vertex in between $u$ and $v$ in $\sigma$. Let $W= \{w: w\leq_{\sigma}u, w \in V(\hat{p})\}$ and $\hat{q}$ is the vector corresponding to $W$, i.e., $V(\hat{q})=W$. Clearly $\hat{q} \perp \hat{p}$. 
Now when $\sigma$ is is restricted to $W$, we get an ordering corresponding  to the vertex set  $V(\hat{q})$.  As the number of $S$-backward arcs with respect to  $\sigma$ with head at $ V(\hat{p}) \smallsetminus V(\hat{q})  $  is $|A^{\gets}(\hat{q}, \hat{p})|$,  the value $\sfas[T(\hat{p})]$ is at least the value of  ${\sfas}[T(\hat{q})]+ |A^{\gets}(\hat{q}, \hat{p})|$,  completing  the  proof.

The Recurrence \ref{equ:onee} naturally leads to a dynamic programming algorithm for the problem. We build a table containing $\sfas[T(\hat{p})]$ for every $ \hat{p} \in \mathscr{Z} $. Our final answer for this problem is the entry $\sfas[T([n_1, n_2, \ldots, n_q])]$.  There are $k^{\OO(q)}$ table entries,	for each entry it takes $k^{\OO(q)}$ time to compute it giving the $k^{\OO(q)}$ time bound.
\end{proof}

\begin{theorem}\label{theo:sfast}
\sfast can be solved in expected time $ 2^{\OO(\sqrt{k} \log k)}+  n^{\OO(1)} $.
\end{theorem}

\begin{proof}
Our algorithm \texttt{ALGO$\_$SFAST} proceeds as described in \cref{fig:algo}. The correctness of the algorithm  follows from	 \cref{lem:DP}. Combining  \cref{obs:goodcoloring}, 	\cref{lem:DP} yields an expected running 	time of $ 2^{\OO(\sqrt{k} \log k)} +  n^{\OO(1)} $ for finding a $ S $-feedback	arc set of size 	at most	$ k $ if one exists.
\end{proof}

\subparagraph{Derandomization.}

For integers $n$, $k$ and $q$, a family $ \mathcal{F} $ of functions from $[n]$ to $ [q] $ is called a {\em universal $ (m, k, r) $-coloring family} if for any graph $ G $ on the vertex set $ [n] $ with at most $k$ edges, there exists a function $ f \in \mathcal{F} $ that properly colors $E(G)$. The follwing fact is known about universal $ (n, k, 2\sqrt{k}) $-coloring family.

\begin{proposition}[Theorem 5.28 \cite{cygan2015parameterized}]\label{prop:uni}
For any $n,k \geq 1$, there exists a universal $ (n, k, \ceil[\big]{2\sqrt{k}}) $-coloring family $ \mathcal{F} $ of size $ 2^{\OO(\sqrt{k} \log k)} \log n$ that can be constructed in time $ 2^{\OO(\sqrt{k} \log k)} n\log n$.
\end{proposition}

Now if we replace the random choices of a coloring in \cref{theo:sfast} with an iteration over the coloring family given by \cref{prop:uni}, we obtain a deterministic subexponential time algorithm for \sfast. 

\sfastheorem*

\section{Conclusion}\label{sec:conclusion}
In this paper, we developed the first quadratic vertex kernel for the \sfast problem. Our kernel used variants of the most well-known reduction rules for {\sc FAST} and introduced two new reduction rules
to identify irrelevant vertices. We believe that the methodology adopted for our kernelization algorithm will also be useful for other problems. 

A non-trivial linear vertex kernel has been established for {\sc FAST}. It would be interesting to design a linear vertex kernel for \sfast that aligns with the bound of {\sc FAST} or to provide a lower bound based on some established conjectures. The same question is also open for {\sc Subset-FVST}.

\bibliography{main}

\end{document}